\documentclass[11pt,english]{article}

\usepackage{amssymb,verbatim}
\usepackage{epsfig}
\usepackage{fancyhdr}
\usepackage{pdflscape}
\usepackage{xcolor}

\usepackage{booktabs}
\usepackage{enumerate}
\usepackage{array}
\usepackage{amstext}
\usepackage{amsthm}
\usepackage{amsopn}
\usepackage{amsfonts}
\usepackage{amsmath}
\usepackage{mathrsfs}
\usepackage{amsbsy}
\usepackage{multicol}
\usepackage{colortbl}
\usepackage{multirow}
\usepackage{fancyhdr}
\usepackage{appendix}
\usepackage{pdflscape}
\usepackage{natbib}
\setcitestyle{authoryear,round}
\usepackage[colorlinks,linkcolor=blue,citecolor=blue]{hyperref}
\usepackage{makecell}

\usepackage{xcolor}

\usepackage{babel}
\usepackage{threeparttable}
\usepackage[section]{placeins}

\renewcommand {\thetable}{\arabic{table}}

\usepackage{color}
\usepackage{soul}
\soulregister\cite7
\soulregister\ref7
\soulregister\citep7
\soulregister\citet7
\soulregister\caption7

\usepackage{setspace}
\renewcommand{\baselinestretch}{1.5}

\theoremstyle{plain}
\newtheorem{theorem}{Theorem}
\newtheorem{lemma}{Lemma}

\newtheorem{proposition}{Proposition}
\usepackage{floatrow}
\newfloatcommand{capbtabbox}{table}[][\FBwidth]

\renewcommand{\baselinestretch} {1.3}
\makeatletter 
\def\singlespace{\def\baselinestretch{1}\@normalsize}

\usepackage{bm}
\usepackage{algpseudocode}
\usepackage{algorithmicx,algorithm}
\newtheorem{definition}{Definition}
\newtheorem{myAss}{Assumption}

\begin{document}
	\title{Causal Mediation Analysis with a Three-Dimensional Image Mediator}
	\author{Minghao Chen and Yingchun Zhou$^{*}$\\
		\footnotesize Key Laboratory of Advanced Theory and Application in Statistics and Data Science-MOE,\\
		\footnotesize Institute of Brain and Education Innovation, School of Statistics,\\
		\footnotesize East China Normal University 3663 North Zhongshan Road, Shanghai, 200062, P.R. China}
	\date{}
	\maketitle
	\newcommand\blfootnote[1]{%
		\begingroup
		\renewcommand\thefootnote{}\footnote{#1}%
		\addtocounter{footnote}{-1}%
		\endgroup
	}
	\blfootnote{*Co-corresponding Authors: Yingchun Zhou, \textit{E-mail address:~yczhou@stat.ecnu.edu.cn}}
	
	\begin{abstract}
		Causal mediation analysis is increasingly abundant in biology, psychology, and epidemiology studies, etc. In particular, with the advent of the big data era, the issue of high-dimensional mediators is becoming more prevalent. In neuroscience, with the widespread application of magnetic resonance technology in the field of brain imaging, studies on image being a mediator emerged. In this study, a novel causal mediation analysis method with a three-dimensional image mediator is proposed. We define the average casual effects under the potential outcome framework, explore several sufficient conditions for the valid identification, and develop techniques for estimation and inference. To verify the effectiveness of the proposed method, a series of simulations under various scenarios is performed. Finally, the proposed method is applied to a study on the causal effect of mother's delivery mode on child's IQ development. It is found that the white matter in certain regions of the frontal-temporal areas has mediating effects. 
	\end{abstract}
	
	\baselineskip18pt
	
	\section{INTRODUCTION}
	Decoding the cognitive functions of the brain is currently one of the most challenging scientific problems in the world \citep{casey2005imaging,Giedd2010structural}. Recently, there has been an increasing interest in augmenting this study with mediation analysis that determines whether the effect of an external exposure on the cognitive or behavioral outcome is mediated by some latent intermediate variables, and estimates these indirect effects \citep{caffo2008brain, lindquist2012functional, zhao2019granger,zhao2021multimodal}. This study presents a novel method that models the magnetic resonance imaging (MRI) data as a three-dimensional mediator to reveal the underlying causal mechanism.
	
	Mediation analysis was early developed in the psychology literature \citep[e.g.,][]{james1984mediators, baron1986moderator}, which frequently investigates the influences of the treatment on the mediator, the mediator on the outcome, and the treatment on the outcome via structural equation models (SEMs), with the model coefficients interpreted as direct and indirect effects. Later on, researches started to discuss the causality of mediation, that is, defining causal parameters of interest, establishing the identifiability assumptions for these parameters and constructing models to estimate or test the mediation effect simultaneously \citep[e.g.,][]{robins1992identifiability,pearl2001direct, rubin2004direct, albert2008mediation, imai2010identification, imai2010general, vanderweele2015explanation}. Causal mediation analysis has been widely used in many scientific areas, and a series of practical frameworks have been developed for different application scenarios. For example, \citet{vanderweele2014mediation} and \citet{daniel2015causal} considered the case of multiple intermediary variables. \citet{lindquist2012functional}, \citet{zhao2019granger} and \cite{jiang2023causal} discussed longitudinal or functional mediators. Furthermore, with the advent of big data era, high-dimensional mediators are becoming more common in causal mediation analysis, and have been studied recently \citep[e.g.,][]{derkach2019high, dai2020multiple, zhou2020estimation, liu2022large, zhang2021mediation,guo2022statistical}. In particular, the widespread application of MRI in the fields of psychology and biology appeals to establishing a framework of image causal mediation analysis (ICMA). However, studies considering image mediators in causal inference are rare.
	
	To assess image mediation, we combine the potential outcome framework of \citet{rubin1974estimating} with linear image structural equation models (LISEMs), and then show that under certain assumptions, it is possible to obtain a valid estimate of the average causal effects of the mediator from the parameters of the LISEMs. Furthermore, in order to determine whether an image has a mediation effect and which loci of the image are transmitting this effect, a test procedure is provided. Although some researchers have considered different strategies for the identification of direct and indirect effects in the cases when interactions and nonlinearities are present \citep[e.g.,][]{avin2005identifiability,petersen2006estimation,imai2010general,vanderweele2015explanation,miles2020semiparametric}, in this paper we are only concerned with the type of additive models that are commonly used in the social sciences and leave other issues for future work.
	
	Additionally, similar to the studies involving high-dimensional variables in causal inference \citep[e.g.,][]{luo2017estimating, zhang2018impact, zhang2020quantile}, regularization for the image coefficient of the second model in the LISEMs are required. While image belongs to high-dimensional data, it is often highly structured, which makes the common dimension reduction methods such as the LASSO method \citep{tibshirani1996regression} less effective \citep[e.g.,][]{zhou2014regularized,wang2017generalized}. Different from two-dimensional images which are considered as matrices, the three-dimensional images in this study are treated as tensors. Thus the method proposed in a recent work considering the LISEMs with two-dimensional image variables \citep{yu2022mapping} is unsuitable for this study, since the tensor nuclear norm regularization approach is computationally intractable \citep{friedland2018nuclear}. To solve the tensor estimation problem in the second model of the LISEMs, the image coeﬀicient is regularized with low Tucker rank \citep{tucker1966some}. Other methods \citep[e.g.,][]{zhou2013tensor,llosa2022reduced} that are applicable to estimating tensors of various structures can also be used in the proposed ICMA method.
	
	This paper is structured as follows. In Section 2, the proposed ICMA method is described, including definition of average casual effects, establishment of valid identifiability assumptions, and description of procedures for estimating parameters and making inference. In Section 3, a series of simulation studies is conducted to verify the performance of the proposed method. In Section 4, the ICMA method is applied to the real data, and some enlightening results are obtained.

	\section{METHODOLOGY}\label{sec:2}
	In this section, we integrate the potential outcome framework of \cite{rubin1974estimating} with SEMs to construct a framework of image causal mediation analysis. The notations are as follows. Let $\mathbb{Z} = \{Z_j\}_{j = 1, \cdots,J}$ denote the set of all pre-treatment confounding variables. Throughout this article, assume $Z_j$'s are all scalars and fully observed with $J \ll n$, where $n$ is the sample size. Let $X$ be the binary treatment variable (0 or 1 corresponding to two treatments). Let $Y$ be the continuous outcome. Denote $\bm{M}$ as the three-dimensional image mediator. In practice, this imaging data is often represented in the form of a three-dimensional tensor. Assume that the tensors in this study are of type $N_1 \times N_2 \times N_3 $, and $N = \prod\limits_{d = 1}^{3}N_d$ may be much larger than $n$. Additionally, let $\mathcal{Z}$, $\bm{\mathcal{M}}$ and $\mathcal{Y}$ denote the support of the distributions of $\mathbb{Z}$, $\bm{M}$ and $Y$, respectively.

	\subsection{Definition and Identification of the Causal Effects}\label{sec:2.1}
	In this section the causal effects of interest are defined and the assumptions for identifiability are established. First, following the potential outcome framework, the potential mediators and potential outcomes are defined as below. Let $\bm M_i(x_1, \cdots, x_n)$ be the potential mediator for unit $i$ under the general treatment status $\{X_l = x_l\}_{l = 1\cdots, n}$, $Y_i(\{x_l\}_{l = 1,\cdots,n}$, $\{\bm M_l(\{x_l\}_{l = 1,\cdots,n})\}_{l = 1,\cdots,n})$ represent the potential outcome for unit $i$ under $\{X_l = x_l\}_{l = 1, \cdots, n}$ and $\{\bm M_l(\{x_l\}_{l = 1,\cdots,n})\}_{l = 1, \cdots, n}$, $Y_i(\{x_l\}_{l = 1,\cdots,n}, \{\bm M_l(\{1-x_l\}_{l = 1,\cdots,n})\}_{l = 1,\cdots,n})$ be the potential outcome for unit $i$ under $\{X_l = x_l\}_{l = 1,\cdots,n}$ and $\{\bm M_l(\{1-x_l\}_{l = 1,\cdots,n})\}_{l = 1,\cdots,n}$, and $Y_i (\{x_l\}_{l = 1,\cdots,n}, \{\bm m_l\}_{l = 1,\cdots,n})$ denote the potential outcome for unit $i$ when $\{X_l = x_l\}_{l = 1,\cdots,n}$ and $\{\bm M_l = \bm m_l\}_{l = 1,\cdots,n}$.
	
	Then, following the literature in causal inference \citep{rubin1978bayesian,rubin1980randomization}, to simplify the models of potential values of mediator and outcome, the stable unit treatment value assumption (SUTVA) that there is no interference between individuals is assumed. The statement below is analogous to that in \citet{angrist1996identification}.
	\begin{myAss}[Stable Unit Treatment Value Assumption]\label{Ass:1}
		\quad \\ 
		$\begin{array}{ll}
			\multicolumn{2}{l}{\text{For each unit i :}}\\
			(a) &\text{If} \ x_i = x_i^{'}, \ \text{then} \ \bm M_i(x_1, \cdots,x_n) = \bm M_i(x_1^{'}, \cdots, x_n^{'}), \\
			&Y_i(\{x_l\}_{l = 1,\cdots,n}, \{\bm M_l(\{x_l\}_{l = 1,\cdots,n})\}_{l = 1,\cdots,n}) = Y_i(\{x_l^{'}\}_{l = 1,\cdots,n}, \{\bm M_l(\{x_l^{'}\}_{l = 1,\cdots,n})\}_{l = 1,\cdots,n}),\\
			&\text{and} \ Y_i(\{x_l\}_{l = 1,\cdots,n}, \{\bm M_l(\{1-x_l\}_{l = 1,\cdots,n})\}_{l = 1,\cdots,n}) = Y_i(\{x_l^{'}\}_{l = 1,\cdots,n}, \{\bm M_l(\{1-x_l^{'}\}_{l = 1,\cdots,n})\}_{l = 1,\cdots,n}).\\
			(b) &\text{If} \ x_i = x_i^{'} \ \text{and} \ \bm m_i = \bm m_i^{'}, \ \text{then} \ Y_i(\{x_l\}_{l = 1,\cdots,n}, \{\bm m_l\}_{l = 1,\cdots,n}) = Y_i(\{x_l^{'}\}_{l = 1,\cdots,n}, \{\bm m_l^{'}\}_{l = 1,\cdots,n}).
		\end{array}$
	\end{myAss}	
	
	SUTVA implies that potential outcomes for each person are unrelated to the treatment given to other individuals. Under SUTVA, let $\bm{M}_i(x)$ denote the potential value of the mediator for unit $i$ under the treatment status $X_i = x$, $Y_i(x,\bm{M}_i(x^{'}))$ represent the potential outcome for unit $i$ under $X_i = x$ and $\bm{M}_i(x^{'})$, and $Y_i (x, \bm m)$ be the potential outcome for unit $i$ when $X_i = x$ and $\bm M_i = \bm m$. Additionally, we make the the consistency assumption that are generally presupposed in causal inference literature as below.
    \begin{myAss}[Consistency Assumption]\label{Ass:2}
			\quad \\
			\text{For each unit i,} $\bm M_i = \bm{M}_i(X_i)$ \text{and} $Y_i = Y_i(X_i,\bm{M}_i(X_i))$.
	\end{myAss}
	
	Under SUTVA, the unit-level causal effects of $X$ on $\bm M$ and $Y$, respectively, for individual $i$ ($i \in \{1,\cdots,n\}$), are given by the expressions
	\begin{equation}
		\bm M_{i}(1)-\bm M_{i}(0),
	\end{equation}
	\begin{equation}\label{eq:2}
		Y_{i}\left(1, \bm{M}_{i}(1)\right)-Y_{i}\left(0, \bm{M}_{i}(0)\right).
	\end{equation}
	
	Then (\ref{eq:2}) can be decomposed as follows.
	\begin{equation}\label{eq:3}
		\begin{aligned}
			Y_{i}(1,&\left.\bm{M}_{i}(1)\right)-Y_{i}\left(0, \bm{M}_{i}(0)\right) \\
			=&\left\{Y_{i}(1, \bm{M}_{i}(0))-Y_{i}(0, \bm{M}_{i}(0))\right\}+\left\{Y_{i}(1, \bm{M}_{i}(1))-Y_{i}(1, \bm{M}_{i}(0))\right\} \\
			=&\left\{Y_{i}(1, \bm{M}_{i}(1))-Y_{i}(0, \bm{M}_{i}(1))\right\}+\left\{Y_{i}(0, \bm{M}_{i}(1))-Y_{i}(0, \bm{M}_{i}(0))\right\} .
		\end{aligned}
	\end{equation}
	
	In both decompositions, the first term represents a natural direct effect of $X$ on $Y$ of subject $i$, while the second represents a natural  indirect effect \citep{pearl2001direct}. 
	
	Averaging (\ref{eq:3}) over all subjects, one obtains
	\begin{equation}\label{eq:4}
		\begin{aligned}
			E&(Y(1,\bm{M}(1)) - Y(0,\bm{M}(0))) \\
			&= E(Y(1,\bm{M}(0)) - Y(0,\bm{M}(0))) + E(Y(1,\bm{M}(1)) - Y(1,\bm{M}(0)))\\
			&= E(Y(1,\bm{M}(1)) - Y(0,\bm{M}(1))) + E(Y(0,\bm{M}(1)) - Y(0,\bm{M}(0))).\\
		\end{aligned}
	\end{equation}
	
	Then the causal effects of interest are defined as follows.
	
	\begin{definition}[Average Total Effect]
		\label{Def:1}	
		\begin{equation}\label{eq:5}
			\tau:= E(Y(1,\bm M(1))-Y(0, \bm M(0))).
		\end{equation}
	\end{definition}
	\begin{definition}[Average Causal Direct Effect]
		\label{Def:2}
		\begin{equation}\label{eq:6}
			\gamma(x):= E(Y(1,\bm M(x))) - E(Y(0,\bm M(x))), 
		\end{equation}	
		for $x = 0,1$.
	\end{definition}
	\begin{definition}[Average Causal Indirect Effect]
		\label{Def:3}
		\begin{equation}\label{eq:7}
			\delta(x):= E(Y(x,\bm M(1))) - E(Y(x,\bm M(0))),
		\end{equation}
		for $x = 0,1$.
	\end{definition}
	
	In general, the average causal direct and indirect effects defined above are different from the average controlled effect of the treatment, that is, $E(Y(x,\bm m) - Y(x',\bm{m}))$ for $x \neq x^{'}$ and all $\bm{m} \in \bm{\mathcal{M}}$, and that of the mediator, that is, $E(Y(x,\bm m) - Y(x,\bm{m}^{'}))$ for $x = 0,1$ and $\bm{m} \neq \bm{m}^{'}$\citep{pearl2001direct,robins2003semantics}. Unlike the causal effects $\gamma(x)$ and $\delta(x)$, the controlled effects are defined in terms of specific values of the mediator, rather than its potential values.
	
	Since for individual $i$ only one of \{$\bm M_i(1)$, $Y_i(1, \bm M_i(1))$\} and \{$\bm M_i(0), Y_i(0, \bm M_i(0))$\} can be observed, the unit-level causal effects cannot be directly calculated. Nevertheless, it will be shown that the average causal effects defined above can be identified and validly estimated under certain assumptions.
	
	\begin{myAss}[Sequential ignorability]
		\label{Ass:3}
		\quad \\
		$\begin{array}{rl}
			(a)&\left\{Y_i(x^{'}, \bm m), \bm{M}_i(x)\right\} \perp X_i \mid \mathbb{Z}_i = z,\\
			(b)&Y_i(x^{'}, \bm m) \perp \bm{M}_i(x) \mid X_i = x, \mathbb{Z}_i = z\\
		\end{array}$\\
		for $x, x^{'}= 0,1$, and all $z \in \mathcal{Z}$, where it is also assumed that $0 < Pr(X = x\mid \mathbb{Z} = z)$ and $0 < Pr(\bm M(x) = m \mid X = x,\mathbb{Z} = z)$ for $x = 0,1$, and all $z \in \mathcal{Z}$ and $\bm{m} \in \bm{\mathcal{M}}$.
	\end{myAss}
	Assumption \ref{Ass:3} is called the sequential ignorability assumption, which has been carefully described and compared with other strong assumptions when the intermediary variable is a one-dimensional variable in \citet{imai2010identification}. In this study, this assumption is extended to the image mediator case. It is easy to show that $\gamma(x)$ and $\delta(x)$ are nonparametrically identified under Assumption \ref{Ass:3}. However, it is not easy to estimate them since estimating the conditional distribution of a tensor-valued random variable is difficult and the estimate often does not have a closed form expression, even if the entries of $\bm{M}$ are independently normally distributed \citep{hoff2011separable}. To this end, seeking sufficient conditions for validly estimating $\tau$, $\gamma(x)$ and $\delta(x)$ is necessary.
	
	\begin{myAss}\label{Ass:4}
		The observed variables of each individual $i \in \{1,\cdots,n\}$ follow the LISEMs, that is,
		\begin{equation}\label{eq:8}
			\bm{M}_i = \bm{\eta}_{1} + X_{i} * \bm{\alpha} + \sum\limits_{j=1}^{J} Z_{ij} * \bm{\Psi}_{j} + \bm{\varepsilon}_{i},
		\end{equation}
		\begin{equation}\label{eq:9}
			Y_{i} = \eta_{2} +X_{i} \gamma + \langle\bm{M}_{i},\bm{\beta}\rangle + \sum\limits_{j=1}^{J} Z_{ij} s_j + \epsilon_{i},
		\end{equation}
		where $E(\bm\varepsilon\mid X = x, \mathbb{Z}=z) = \bm 0$, $E(\epsilon \mid X=x, \bm M = \bm m, \mathbb Z = z) = 0$, $c*\bm{U} = \{c u_{p_1,p_2,p_3}\}^{1\leq p_d\leq N_d}_{d = 1,2,3}$ denotes the scalar multiplication between the scalar $c$ and tensor $\bm U = \{u_{p_1,p_2,p_3}\}^{1\leq p_d \leq N_d}_{d = 1,2,3}$, and $\langle\bm U,\bm V\rangle = \sum_{p_1,p_2,p_3} u_{p_1, p_2, p_3} v_{p_1,p_2,p_3}$ denotes the inner product between the tensors $\bm{U} =  \{u_{p_1,p_2,p_3}\}^{1\leq p_d \leq N_d}_{d = 1,2,3}$ and $\bm{V} =  \{v_{p_1,p_2,p_3}\}^{1\leq p_d \leq N_d}_{d = 1,2,3}$. 
	\end{myAss}
	
	Assumption \ref{Ass:4} extends the linear structural equation models in \citet{baron1986moderator} to LISEMs by mainly replacing $M_i$ with $\bm M_i$, and allows that the observed variables follow the LISEMs. Under Assumptions \ref{Ass:1}-\ref{Ass:4}, one can obtain the average causal effects from the parameters of the LISEMs, which is presented in the following theorem.
	
	\begin{theorem}\label{th:1}
		If Assumptions $\ref{Ass:1}$-$\ref{Ass:4}$ hold, then the causal effects $\tau$, $\gamma(x)$ and $\delta(x)$ can be expressed as
		\begin{equation*}
			\begin{array}{ll}
				&\tau = \gamma + \langle\bm{\beta},\bm \alpha\rangle,\\
				&\gamma(x) = \gamma,\\
				&\delta(x) = \langle\bm{\beta},\bm \alpha\rangle,
			\end{array}
		\end{equation*}
		for $x = 0,1$.
	\end{theorem}
	See Appendix A for a proof. Theorem \ref{th:1} implies that there is no interaction between the treatment and the mediator, that is, $\gamma(0) = \gamma(1)$ and $\delta(0) = \delta(1)$. 
	
	Alternatively, inspired by the relationship between the average controlled causal effects and the average natural causal effects \citep[e.g.,][]{zheng2015causal}, we find some conditions that are weaker than Assumptions \ref{Ass:3} and \ref{Ass:4}, but still sufficient to identify $\tau$, $\gamma(x)$ and $\delta(x)$. These conditions are summarized in Assumption \ref{Ass:5}.
	
	\begin{myAss}\label{Ass:5}
		\quad \\
		$\begin{array}{rl}
			(a) &E\{\bm{M}(x) \mid X=x, \mathbb{Z} = z\} = E\{\bm{M}(x) \mid \mathbb{Z} = z\}, \\
			&E\{Y(x,\bm{M}(x)) \mid X = x, \bm{M}(x) = m, \mathbb{Z} = z\} = E\{Y(x, \bm{M}(x)) \mid \bm{M}(x) = m, \mathbb{Z} = z\}, \\
			(b)&E\{Y(x,\bm m) \mid \bm{M}(x') = \bm m, \mathbb{Z} = z\} = E\{Y(x,\bm m) \mid \mathbb{Z} = z\},\\
			(c) &0 < Pr(X = x\mid \mathbb{Z}_i = z) \ \text{and} \ 0 < p(\bm M(x) = m \mid \mathbb{Z} = z),\\
			(d)&E\{\bm{M}(x) \mid \mathbb{Z} = z\} = \bm{\eta}_1^{(c)} + x * \bm{\alpha}^{(c)}+ \sum\limits_{j=1}^{J}z_j * \bm{\Psi}_j^{(c)},\\
			&E\{Y(x,\bm m) \mid \mathbb{Z} = z\} = \eta_2^{(c)} + x\gamma^{(c)} + \langle\bm{m}, \bm{\beta}^{(c)}\rangle + \sum\limits_{j=1}^{J}z_j s_j^{(c)}
		\end{array}$\\
		for $x = 0,1$, and all $z \in \mathcal{Z}$ and $\bm m \in \mathcal{M}$.
	\end{myAss}
	
	Under Assumptions \ref{Ass:1} and \ref{Ass:5} (d), the average controlled direct effect of the treatment is defined as $E(Y(1, \bm m) - Y(0, \bm m)) = \gamma^{(c)}$ and the average controlled effect of $X$ on $\bm M$, and $\bm M$ on $Y$ at level $\bm m$ versus $\bm{m}^{\prime}$ can be written as $E(\bm{M}(1)-\bm{M}(0)) = \bm{\alpha}^{(c)}$, $E(Y(x, \bm m) - Y(x, \bm m^{\prime})) = \langle \bm{m} - \bm{m}^{\prime}, \bm{\beta}^{(c)}\rangle$, respectively. Further, under Assumptions \ref{Ass:1}, \ref{Ass:2} and \ref{Ass:5} (a)-(c), the parameters $\bm{\alpha}^{(c)} = \bm{\alpha}$, $\gamma^{(c)} = \gamma$ and $\bm{\beta}^{(c)} = \bm{\beta}$, which in general do not hold \citep{sobel2008identification,lindquist2012functional}. Then the average controlled causal effects can be expressed by the parameters of the LISEMs. It follows that $\tau$, $\gamma(x)$ and $\delta(x)$ can be validly identified under Assumptions \ref{Ass:1}, \ref{Ass:2} and \ref{Ass:5}. These results are summarized in Theorem \ref{th:2}.
	\begin{theorem}\label{th:2}
		Under Assumptions $\ref{Ass:1}$, $\ref{Ass:2}$ and $\ref{Ass:5}$,
		\begin{equation*}
			\begin{aligned}
				&\bm{\alpha}^{(c)} = \bm{\alpha}, \ \gamma^{(c)} = \gamma, \ \bm{\beta}^{(c)} = \bm{\beta}, \ \text{and}\\
				&E(\bm M(1) - \bm M(0)) = \bm{\alpha},\\
				&E(Y(1, \bm m) - Y(0, \bm m)) = \gamma,\\
				&E(Y(x, \bm m) - Y(x, \bm m^{\prime})) = \langle \bm{m}-\bm{m}^{\prime}, \bm \beta \rangle,\\
				&\tau = \gamma + \langle\bm{\beta},\bm \alpha\rangle, \gamma(x) = \gamma, \delta(x) = \langle\bm{\beta},\bm \alpha\rangle,
			\end{aligned}
		\end{equation*}
		for $x = 0,1$, and all $\bm m \in \mathcal{M}$.
	\end{theorem}
	See Appendix A for the proof of Theorem \ref{th:2}, and the proof that Assumption $\ref{Ass:5}$ is weaker than Assumption $\ref{Ass:3}$ and $\ref{Ass:4}$.
	
	\subsection{Estimation and Inference}\label{sec:2.2}
	In this section an approach for estimating the causal parameters of interest in the LISEMs is described. Additionally, a procedure for testing the causal parameters based on the resampling techniques is provided.
	
	\subsubsection{Estimation}\label{sec:2.2.1}
	The LISEMs considered in this study are composed of two types of linear regression models. One consists of an image response and some scalar predictors, while the other consists of a scalar response, an image predictor, and some other scalar predictors. Since the second one is essentially a special high-dimensional linear model, it is assumed that the Tucker rank of the coefficient tensor $\bm{\beta}$ is low so as to obtain a valid estimate, as well as to enhance power of the subsequent tests. Note that the proposed method is flexible, so other regularization methods that fit image data may also be used. Meanwhile, for further statistical inference, not penalizing $\bm{\alpha}$, $\bm{\Psi}_j$ and $\tau$ leaves their estimators unbiased, and there is no need for conducting any of the debiased or decorrelated procedures \citep[e.g.,][]{zhang2014confidence,van2014asymptotically}. Our estimation strategy is as follows.
	
	For simplicity, the intercept terms is suppressed in the LISEMs, 
	\begin{equation}\label{eq:10}
		\bm{M}_{i} = X_{i} * \bm{\alpha} + \sum\limits_{j=1}^{J}Z_{ij} * \bm{\Psi}_j + \bm{\varepsilon}_{i},
	\end{equation}
	\begin{equation}\label{eq:11}	
		Y_{i} = X_{i} \gamma + \langle\bm{M}_{i},\bm{\beta} \rangle + \sum\limits_{j=1}^{J}Z_{ij} s_j + \epsilon_{i}.
	\end{equation}
	
	First, an unbiased estimate of $\bm \Delta = (\bm \alpha, \bm \Psi_1, \cdots, \bm \Psi_J)$ can be obtained via the least square method based on (\ref{eq:10}), that is, 
	\[
	\begin{array}{rl}
		\widehat{\bm \Delta} &= \underset{\bm \alpha,\bm \Psi_{j} \in \mathbb{R}^{N_{1} \times N_{2} \times N_{3}}}{\arg\min}{\sum\limits_{i = 1}^{n}\parallel\bm M_i - X_i * \bm \alpha - Z_{ij}*\bm{\Psi}_{j}\parallel_{F}^2}\\
		&=  \underset{\bm \alpha,\bm \Psi_{j} \in \mathbb{R}^{N_{1} \times N_{2} \times N_{3}}}{\arg\min}\parallel M  - A_{X,Z} \Delta \parallel_{2}^2,
	\end{array}
	\]
	where $|| \bm U ||_{F}^{2} = \sum_{p_1,p_2,p_3} u^2_{p_1,p_2,p_3}$ denotes the Frobenius norm, $M = \left(\mathrm{vec}(\bm M_1)^{T},\cdots,\mathrm{vec}(\bm M_n)^{T}\right)^{T}$ and $\mathrm{vec}(\bm U) = (u_{1,1,1}, u_{2,1,1},\cdots, u_{N_1,1,1}, u_{1,2,1},\cdots, u_{N_1,2,1}, \cdots, u_{N_1,N_2,N_3})^{T}$. Additionally, $A_{X,Z} = (A_X,A_{Z_1}, \cdots,A_{Z_J})$, $A_X = \left( X_1 I_{N}, \cdots, X_n I_{N} \right)^{T}$, $A_{Z_j} = \left( Z_{1j} I_{N}, \cdots, Z_{nj} I_{N} \right)^{T}$, $\Delta = (\alpha, \Psi_1, \cdots, \Psi_J)$, $\alpha = \mathrm{vec}(\bm \alpha)$, $\Psi_j= \mathrm{vec}(\bm \Psi_j)$, $j \in \{1,\cdots,J\}$.
	
	Next, $\bm \beta$ is estimated based on equation $(\ref{eq:11})$, which is a problem with rising interest in statistics. Specifically, inspired by recent work \citep[e.g.,][]{li2018tucker,ahmed2020tensor,han2022optimal,luo2021low}, assuming the parameter $\bm \beta$ to have a low Tucker rank, we use the Tucker decomposition to solve the tensor estimation problem. In this study, an estimate of $\gamma$ and a valid estimate of $\bm\beta$ are obtained via the method proposed by \cite{li2018tucker}, in which the tensor estimator of $\bm \beta$, actually $\mathrm{vec}(\bm \beta)$, has a normal distribution limit under certain conditions, such as $\epsilon$ is normal and the parameter space of $\bm \beta$ is compact and contains the true value. Moreover, this asymptotic property of $\widehat{\bm \beta}$ satisfies the requirement for the bootstrap step in the inference. See Appendix B for a brief algorithm and selection of Tucker order for $\bm{\beta}$.

	Finally, by substituting $(\ref{eq:10})$ into $(\ref{eq:11})$, one obtains that
	\begin{equation}\label{eq:12}
		Y_{i} = X_{i}\tau + \sum\limits_{j=1}^{J}Z_{ij} \theta_j + \xi_i,
	\end{equation}
	where $\tau = \gamma + \langle\bm{\beta} , \bm \alpha \rangle$ is the total effect of interest, $\theta_j = s_j + \langle \bm{\beta}, \bm{\Psi_j} \rangle$, $\xi_i = \langle\bm \beta, \bm{\varepsilon}_i\rangle + \epsilon_{i}$, and $E(\xi \mid X = x) = 0$. Next, an estimate of $\theta = (\tau,\theta_1, \cdots, \theta_J)^T$ via the ordinary least square method based on $(\ref{eq:12})$ is as follows.
	\[
	\begin{array}{lc}
		&\hat{\theta} = \underset{\theta \in \mathcal{R}^{J+1}}{\arg\min}{\sum_{i = 1}^{n}(Y_i - X_i \theta_0 - Z_{ij} \theta_j)^2}\\
		&\ \ = \underset{\theta \in \mathcal{R}^{J+1}}{\arg\min}{ \parallel Y - B_{X,Z}\theta  \parallel_{2}^2},
	\end{array}
	\]
	where $Y = (Y_1,\cdots,Y_n)^{T}$, $B_{X,Z} = (X, Z_1, \cdots, Z_J)$, $X = (X_1,\cdots,X_n)^{T}$ and $Z_j = (Z_{1j}, \cdots, Z_{nj})^T$.

	
	

	\subsubsection{Inference}\label{sec:2.2.2}
	Another central question of mediation analysis is to determine whether $\bm M$ acts as an intermediary to transmit the effect of $X$ on $Y$, and which loci in $\bm M$ transmit that effect. Analogous to the univariate setting \citep{sobel1982asymptotic, baron1986moderator, lindquist2012functional}, these can be done by testing whether $\langle\bm \alpha, \bm \beta \rangle$ differs from 0, and whether $\bm \alpha \circ \bm \beta$ (the Hadamard product \citep{kressner2017recompression} for $\bm \alpha$ and $\bm \beta$) differs from $\bm{0}$.  Since the exact distribution of $\bm \alpha\circ \bm \beta$ is unknown, bootstrap methods are used to perform inference. The procedure for obtaining the bootstrap distributions of $\bm \alpha$ and $\bm \beta$ is as follows.
	
	\begin{description}
		\item[Step 1.] Estimating the parameters of interest based on the LISEMs and the original data via the strategy proposed in Section \ref{sec:2.2.1}.
		\item[Step 2.] Independently generating a bootstrap sample $(\mathbb{Z}_i, X_i,\bm{M}_{i}^{\dag},Y_i^{\dag})$, $i \in \{1,\cdots,n\}$ based on LISEMs, that is,
		\[
		\bm{M}_{i}^{\dag}=\widehat{\bm \eta}_{1}+ X_{i} * \widehat{\bm{\alpha}} + \sum\limits_{j=1}^{J}Z_{ij} * \widehat{\bm{\Psi}}_j + v_{i}^{\dag} \widehat{\bm{\varepsilon}}_{i},
		\]
		\[
		Y_{i}^{\dag}=\hat{\eta}_{2}+X_{i} \hat{\delta}+\langle\widehat{\bm{\beta}}, \bm{M}_{i}^{\dag}\rangle+\sum\limits_{j=1}^{J}Z_{ij} \hat{s}_j +v_{i}^{\dag} \hat{\epsilon}_{i},
		\]
		where $\widehat{\bm{\varepsilon}}_{i} = \bm M_i - \widehat{\bm{\eta}}_{1}- X_{i} * \widehat{\bm{\alpha}} - \sum\limits_{j=1}^{J}Z_{ij} * \widehat{\bm{\Psi}}_j$, $\hat{\epsilon_{i}} =Y_i - \hat{\eta}_2 - X_{i}\widehat{\delta} - \langle\widehat{\bm{\beta}},\bm{M}_i\rangle-\sum\limits_{j=1}^{J}Z_{ij} \hat{s}_j$, $v_i^{\dag}$ are independent and identically distributed and $Pr(v_i^{\dag} = 1)=Pr(v_i^{\dag} = -1) = \frac{1}{2}$.
		\item[Step 3.] Estimating the parameters of interest again based on the LISEMs and the bootstrap samples generated in Step 2.
		\item[Step 4.] Repeat the procedure outlined in Steps 2 and 3 for a prespecified number of times (e.g., 500 times). These replications are used to compute the bootstrap distributions of $\bm \alpha$ and $ \bm \beta$.
	\end{description}
	
	Next, based on the bootstrap distributions of $\bm \alpha$ and $ \bm \beta$, a MaxP test \citep{mackinnon2002comparison} can be performed to determine which $\alpha_{p_1,p_2,p_3}\beta_{p_1,p_2,p_3}$ is significantly different from 0, and the $p$-value for each locus $(p_1,p_2,p_3)$ equals to $\max\{2- 2\Phi(|\frac{\widehat{\alpha}_{p_1,p_2,p_3}}{\widehat{\sigma}_{\alpha_{p_1,p_2,p_3}}}|), 2- 2\Phi(|\frac{\widehat{\beta}_{p_1,p_2,p_3}}{\widehat{\sigma}_{\beta_{p_1,p_2,p_3}}}|)\}$, where $\Phi(\cdot)$ is the cumulative distribution function of the standard normal distribution, $\widehat{\sigma}_{\alpha_{p_1,p_2,p_3}}$ and $\widehat{\sigma}_{\beta_{p_1,p_2,p_3}}$ are the estimates of standard deviation calculated by the bootstrap distributions.
	
	As we are interested in testing whether $\bm \alpha \circ \bm \beta$ differs from $\bm{0}$ at each locus, it is necessary to adjust the $p$-values for multiple comparisons. Throughout the article, the Benjamini–Hochberg procedure \citep{benjamini1995controlling} is used to control the false discovery rate (FDR). 
	
	\section{SIMULATION}
	In this section, simulation studies are conducted and the performance of the ICMA method is shown in various scenarios. 
	
	The data $(Z_i, X_i, \bm{M}_i, Y_i)$, $i=1,\cdots,n$, $n=100$, $N_1 = N_2 = N_3 = 8$ are generated based on five causal diagrams (five scenarios) which commonly occur in causal mediation analysis. The first three scenarios do not have mediation effect, that is, the treatment has no effect on the mediator or the mediator has no effect on the outcome. The remaining two scenarios are set to have a mediation effect. The specific settings of the five scenarios are summarized in Table \ref{ta:1}.
	
	\begin{table}[htbp]
		\caption{The summary of simulation settings.}
		\begin{tabular*}{\columnwidth}{c|c|c|c|c}
			\hline 
			Scenario & $X$ & $\bm M$ & $Y$ & $\delta$-effect in each locus\\ 
			\hline  1 & $0 \ \text{or} \ 1$ & $Z_i*\bm\Psi + \bm{\varepsilon}_i$ & $X_i \gamma + Z_i s + \epsilon_{i}$ &None\\ 
			\hline  2 & $0 \ \text{or} \ 1$ & $X_i*\bm{\alpha} + Z_i*\bm\Psi  + \bm{\varepsilon}_i$ & $X_i \gamma + Z_i s + \epsilon_{i}$ &None\\ 
			\hline  3 & $0 \ \text{or} \ 1$ & $Z_i*\bm\Psi + \bm{\varepsilon}_i$ & $Z_i s + \langle \bm{M}_i, \bm{\beta} \rangle + \epsilon_{i}$ &None\\ 
			\hline  4 & $0 \ \text{or} \ 1$ & $X_i*\bm{\alpha} + Z_i*\bm\Psi + \bm{\varepsilon}_i$ & $Z_i s + \langle \bm{M}_i, \bm{\beta} \rangle + \epsilon_{i}$ &$\bm{\alpha} \circ \bm{\beta}$\\ 
			\hline  5 & $0 \ \text{or} \ 1$ &  $X_i*\bm{\alpha} + Z_i*\bm\Psi + \bm{\varepsilon}_i$ & $X_i \gamma + Z_i s + \langle \bm{M}_i, \bm{\beta} \rangle + \epsilon_{i}$ &$\bm{\alpha} \circ \bm{\beta}$\\ 
			\hline \multicolumn{5}{l}{\quad Here, $\gamma = s = 10$, $Z_i \sim \mathcal{U}(0,2)$, $X_i \sim \mathcal{B}(1,\pi_i)$, $\text{logit}(\pi_i) = 0.5 - 0.5Z_i$, $\mathrm{vec}(\bm{\varepsilon}_i) \sim \mathcal{N}(0,0.3^2I)$, } \\
			\multicolumn{5}{l}{\quad $\epsilon_i \sim \mathcal{N}(0,0.1^2)$. $\bm{\Psi}$, $\bm \alpha$ and $\bm \beta$ are shown in Figure \ref{fig:S.1}.}\\
			\hline
		\end{tabular*} 
		\label{ta:1}
	\end{table}
	Each of the five scenarios has 500 replications. For each replication, the LISEMs are fitted using the approach outlined in Section \ref{sec:2}. After estimation, a MaxP test based on the generated bootstrap sample is performed to determine which $\alpha_{p_1,p_2,p_3}\beta_{p_1,p_2,p_3}$ is significantly different from 0. The results are controlled for multiple comparisons using the Benjamini-Hochberg procedure ({\it q} = 0.05).
	
	\paragraph{Scenario 1} As shown in Table \ref{ta:1}, the data are generated assuming $ \bm M_i =Z_i*\bm\Psi + \bm \varepsilon_i $ and $Y_i = X_i \gamma + Z_i s + \epsilon_{i} $ for $i = 1, \cdots, n$, where $\mathrm{vec}(\bm \varepsilon_i) \sim \mathcal{N}(0,0.3^2I)$, $\gamma=10$, $Z_i \sim \mathcal{U}(0,2)$, $X_i \sim \mathcal{B}(1,\pi_i)$, $\text{logit}(\pi_i) = 0.5 - 0.5Z_i$, and $\epsilon_i \sim \mathcal{N}(0,0.1^2)$. In this scenario, $X$ has a significant direct effect on $Y$ but no indirect effect, that is, $\gamma = 10$ and $\bm{\alpha} = \bm{\beta} = \bm{0}$. The estimate of $\gamma$ is $\hat{\gamma} = 10.0005$ and its confident interval is $[9.9574,10.0437]$. The first two rows of Figure \ref{fig:1} show estimates of $\bm{\alpha}$, $\bm{\beta}$, and $\bm{\alpha} \circ\bm{\beta}$, together with 95\% confidence intervals of the corresponding straightened vectors. None of the effects appear to deviate significantly from 0 for all the loci. Figure \ref{fig:2}(a) shows the proportion of times that $\bm{\alpha} \circ \bm{\beta}$ is statistically significant (adjusted {\it p} <0.05) among the 500 replications for each locus. Clearly, the proportions for all loci fall well below 0.05, which shows that the estimation and testing for $\bm{\alpha} \circ \bm{\beta}$ perform well. Additionally, as shown in Figure \ref{fig:S.2}(a), even without multiple test correction, the proportions for all loci fall well below 0.05.
	
	\paragraph{Scenario 2} The data are generated in an analogous manner as in Scenario 1, except that $ \bm M_i = X_i*\bm{\alpha} + Z_i*\bm\Psi  + \bm{\varepsilon}_i $ for $i = 1, \cdots, n$, where $Z_i$, $X_i$, and $\bm \varepsilon_i$ are defined as above. Again, $X$ has a significant direct effect on $Y$ but no indirect effect, that is, $\gamma = 10$ and $\bm{\beta} = \bm{0}$. Note that in this case, $\bm \alpha \neq \bm{0}$, which is different from the setting of Scenario 1. The estimate of $\gamma$ is $\hat{\gamma} = 10.0004$ and its confident interval is $[9.7699,10.2309]$. The third and fourth rows of Figure \ref{fig:1} show estimates of $\bm{\alpha}$, $\bm{\beta}$, and $\bm{\alpha} \circ \bm{\beta}$, together with 95\% confidence intervals of the corresponding straightened vectors. The estimates coincide with the true values, with only $\bm \alpha$ being significantly different from $\bm 0$ in the loci corresponding to the elements in $\bm \alpha$ which are nonzero. Figure \ref{fig:2}(b) shows the proportion of times that $\bm{\alpha} \circ \bm{\beta}$ is statistically significant (adjusted {\it p}<0.05) among the 500 replications for each locus. Again, all loci fall well below 0.05. Additionally, in Figure \ref{fig:S.2}(b), the proportions for all loci fall well below 0.05 without multiple test correction. Note that the proportion of false positives increases and approaches 0.05 in the loci corresponding to the elements in $\bm \alpha$ which are nonzero.
	
	\paragraph{Scenario 3} The data are generated assuming $ \bm M_i =Z_i*\bm\Psi + \bm{\varepsilon}_i$ and $Y_i = Z_i s + \langle \bm{M}_i, \bm{\beta} \rangle + \epsilon_{i}$ for $i = 1, \cdots, n$, where $Z_i$ is defined as above, both $\mathrm{vec}(\bm{\varepsilon}_{i})$ and $\epsilon_{i}$ follow normal distributions as described before. In this scenario, $X$ has no direct or indirect effect on $Y$, that is, $\gamma = 0$ and $\bm{\alpha} = \bm{0}$. The estimate of $\gamma$ is $\hat{\gamma} = -0.00001$ and its confidence interval is $[-0.0385, 0.0385]$. The fifth and sixth rows of Figure \ref{fig:1} show estimates of $\bm{\alpha}$, $\bm{\beta}$, and $\bm{\alpha} \circ \bm{\beta}$, together with 95\% confidence intervals of the corresponding straightened vectors. Again, the estimates coincide with the true values, with $\bm \beta$ being significantly different from $\bm 0$ in the loci corresponding to the elements that $\bm \beta$ are nonzero. Figure \ref{fig:2}(c) shows the proportion of times that $\bm{\alpha} \circ \bm{\beta}$ is statistically significant (adjusted {\it p}<0.05) among the 500 replications for each locus. Again, for all loci the proportion falls below 0.05. However, different from the first two scenarios, in this case without multiple test correction, the proportion of false positives increases and exceeds 0.05 (still below 0.1) in the loci corresponding to the elements in $\bm \beta$ which are nonzero, which is presented in Figure \ref{fig:S.2}(c).
	
	\begin{figure}[htbp]
		\centering
		\includegraphics[width=0.85\linewidth]{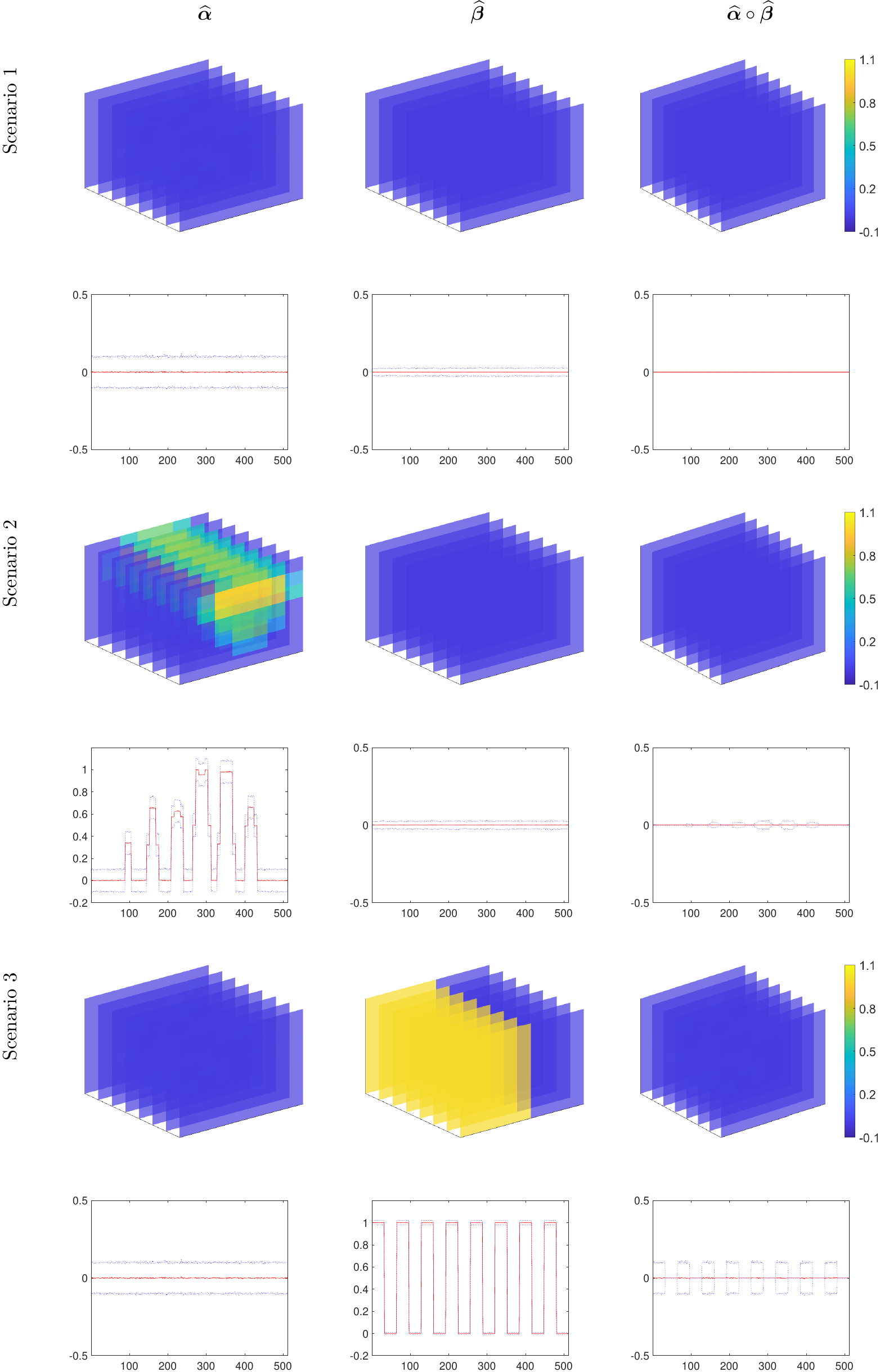}
		\caption{Estimates of $\bm{\alpha}$, $\bm{\beta}$, and $\bm{\alpha} \circ \bm{\beta}$, together with 95\% confidence intervals of the corresponding straightened vectors, for data generated according to the settings described in the Scenarios 1, 2, and 3. The $x$-axis of rows 2, 4 and 6 all represents the straightened loci, and the $y$-axis represents the values of the corresponding parameters at these loci. The colors in the images represent the grid point values.
		\label{fig:1}}
	\end{figure}
	\begin{figure}[htbp]
		\centering
		\includegraphics[width=1\linewidth]{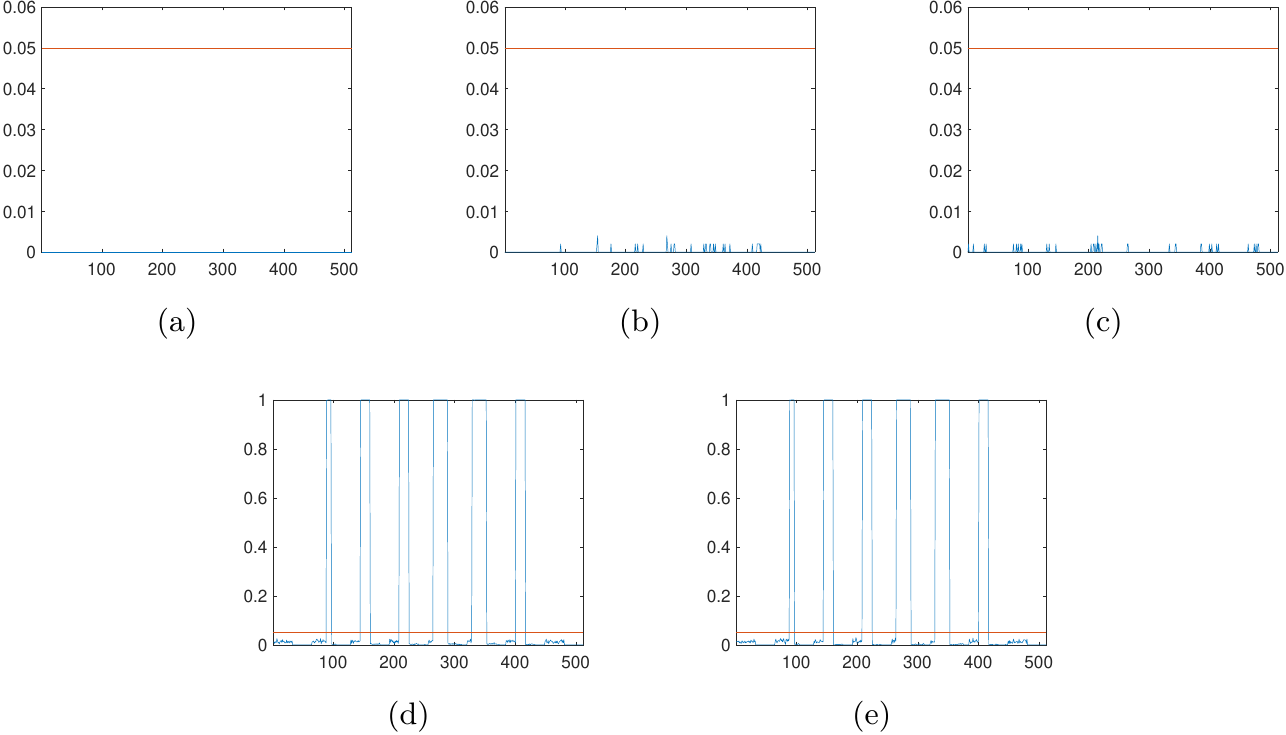}
		\caption{Results of the first three scenarios are shown in (a), (b), and (c), respectively. Each plot shows the proportion of times that $\bm{\alpha} \circ \bm{\beta}$ is statistically significant (adjusted ${\it p}<0.05$) among the 500 replications for each locus. The results illustrate that the method provides adequate control of the false positive rate in all three simulated scenarios. Results of the last two scenarios are shown in (d) and (e), respectively. Each plot illustrates the power of the method in detecting true positives in the loci of $\alpha_{p_1,p_2,p_3} \beta_{p_1,p_2,p_3} \neq 0$, while accurately controlling for false positives in the loci of $\alpha_{p_1,p_2,p_3}\beta_{p_1,p_2,p_3} = 0$. The $x$-axis of all plots represents the straightened loci.}
		\label{fig:2}
	\end{figure}
	
	\paragraph{Scenario 4} The data are generated in an analogous manner as described in Scenario 3, except that $ \bm M_i = X_i*\bm{\alpha} + Z_i*\bm\Psi + \bm{\varepsilon}_i$ for $i = 1, \cdots, n$, where $Z_i$, $X_i$, and $ \mathrm{vec}(\bm{\varepsilon}_{i}) $ are defined as above. In this scenario, $X$ has an indirect effect on $Y$ but no significant direct effect, and some loci in $\bm M$ transmit the effect of $X$ on $Y$, that is, $\gamma = 0$ and $\bm{\alpha} \circ \bm{\beta} \neq \bm{0}$. The estimate of $\gamma$ is $\hat{\gamma} = 0.0712$ and its confidence interval is $[-0.5754, 0.7179]$. The first two rows of Figure \ref{fig:3} show estimates of $\bm{\alpha}$, $\bm{\beta}$ and $\bm{\alpha} \circ\bm{\beta}$, together with 95\% confidence intervals of the corresponding straightened vectors. Again, all three estimates coincide with the true values, as they are significantly different from 0 in the appropriate loci. Figure \ref{fig:2}(d) illustrates the power of the method in detecting true positives in the loci where $\alpha_{p_1,p_2,p_3} \beta_{p_1,p_2,p_3} \neq 0$, while accurately controlling for false positives in the loci where $\alpha_{p_1,p_2,p_3}\beta_{p_1,p_2,p_3} = 0$.  Additionally, the slight difference between Figure \ref{fig:2}(d) and Figure \ref{fig:S.2}(d) is that the latter appropriately controls for false positives in the loci where $\alpha_{p_1,p_2,p_3}\beta_{p_1,p_2,p_3} = 0$ without multiple test correction.
	
	\paragraph{Scenario 5} The data are generated in an analogous manner as described in Scenario 4, except that $Y_i = X_i \gamma + Z_i s + \langle \bm{M}_i, \bm{\beta} \rangle + \epsilon_{i}$ for $i = 1, \cdots, n$, where $\gamma = 10$, $Z_i \sim \mathcal{U}(0,2)$, $\text{logit}(\pi_i) = 0.5 - 0.5Z_i$, $X_i \sim \mathcal{B}(1,\pi_i)$, and $\epsilon_i \sim \mathcal{N}(0,0.1^2)$. In this scenario, $X$ has both direct and indirect effect on $Y$, and some loci in $\bm M$ transmit the effect of $X$ on $Y$, that is, $\gamma = 10$ and $\bm{\alpha} \circ\bm{\beta} \neq \bm{0}$. The estimate of $\gamma$ is $\hat{\gamma} = 10.0033$ and its confident interval is $[9.3608, 10.6457]$. The third and fourth rows of Figure \ref{fig:3} show estimates of $\bm{\alpha}$, $\bm{\beta}$ and $\bm{\alpha}\circ\bm{\beta}$, together with 95\% confidence intervals of the corresponding straightened vectors. Similarly, all three estimates coincide with the true values, as they are significantly different from 0 in the appropriate loci. Figure \ref{fig:2}(e) illustrates the power of the method in detecting true positives in the intervals that $\alpha_{p_1,p_2,p_3} \beta_{p_1,p_2,p_3} \neq 0$, while accurately controlling for false positives in the intervals that $\alpha_{p_1,p_2,p_3}\beta_{p_1,p_2,p_3} = 0$. Additionally, the difference between Figure \ref{fig:2}(e) and Figure \ref{fig:S.2}(e) is that the latter appropriately controls for false positives in the loci where $\alpha_{p_1,p_2,p_3}\beta_{p_1,p_2,p_3} = 0$ without multiple test correction.
	\begin{figure}
		\centering
		\includegraphics[width=0.85\linewidth]{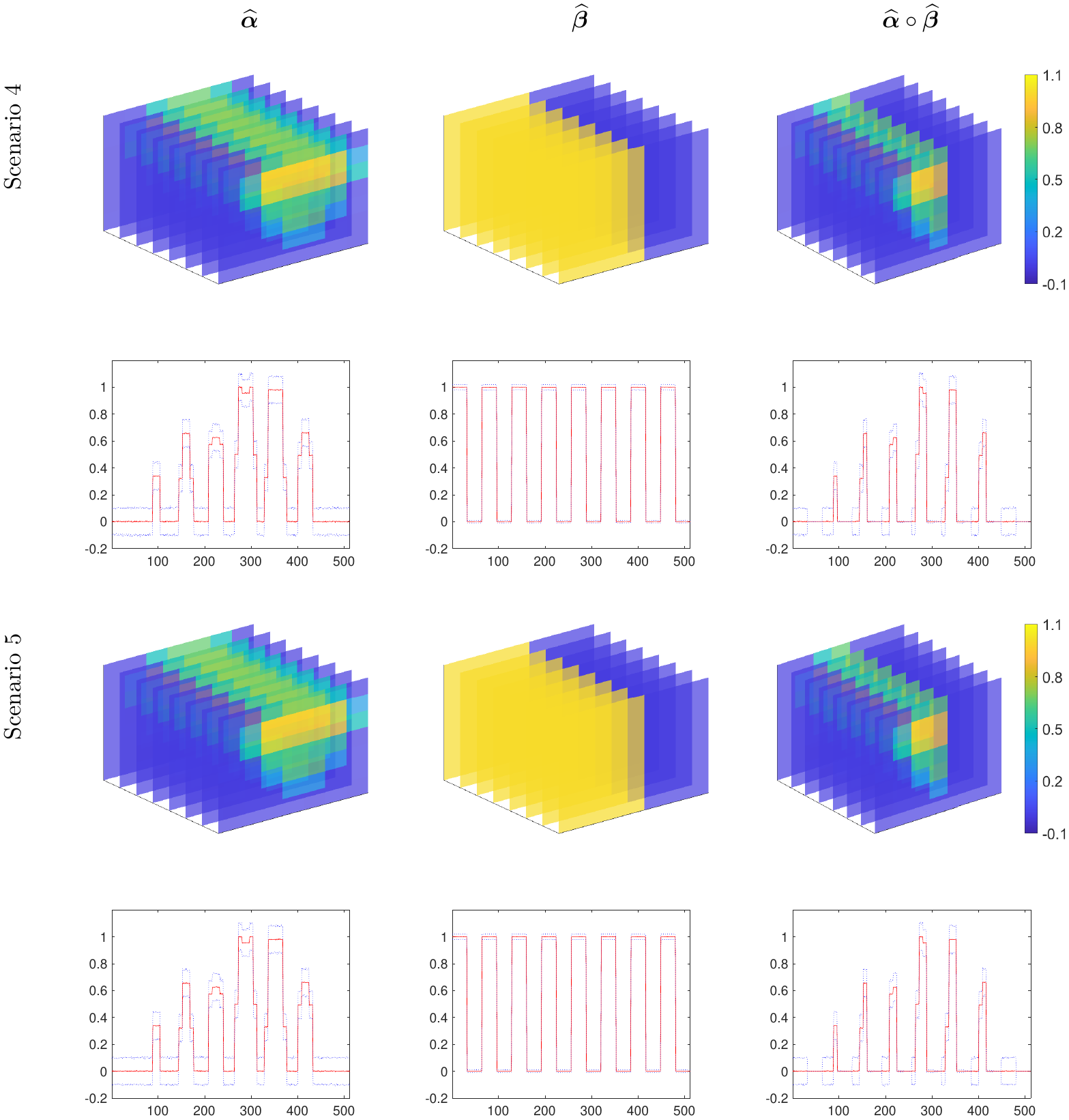}
		\caption{Estimates of $\bm{\alpha}$, $\bm{\beta}$, and $\bm{\alpha} \circ \bm{\beta}$, together with 95\% confidence intervals of the corresponding straightened vectors, for data generated according to the settings described in Scenarios 4 and 5. The $x$-axis represents the straightened loci, and the $y$-axis represents the values of the corresponding parameters at these loci. The colors in the images represent the grid point values.}
		\label{fig:3}
	\end{figure}
	
	\section{REAL DATA APPLICATION}\label{sec:4}
	Extensive scientific research have shown that the brain, as the carrier of human consciousness, profoundly affects the performance of human intelligence \citep[e.g.,][]{reiss1996brain,haier2004structural,northam2011total,liu2019neuronal}.
	
	Recently, a systematic review and meta-analysis by \citet{zhang2019association} indicates that children delivered through caesarean section are associated with an increased risk of negative performance in physical health and neurocognitive development, compared with children born via vaginal delivery. \citet{blazkova2020impact} shows that children born with caesarean section score lower in psychological cognitive tests than those born via vaginal delivery. However, an earlier study \citep{li2011cesarean} on Chinese children finds no difference in intelligence quotient (IQ) scores between children who were delivered by caesarean section and spontaneous vaginal delivery. A more recent systematic review \citep{blake2021association} concludes that the evidence of an association between caesarean section birth and lower offspring cognitive functioning is inconsistent. 
	
	Motivated by the controversial studies, this article considers two issues. The first is whether mother's delivery mode has an impact on child's intellectual performance, and the second is if this impact exists, whether it is mediated by the brain and which loci transmit the impact. These two issues can be tackled with the proposed ICMA method through estimating the causal parameters $\bm\alpha$, $\gamma$, $\bm\beta$, and $\tau$.
	
	The real data are collected by the Institute of Brain Science, East China Normal University. The original dataset contains each child's family status, physical and mental status, sleep duration and quality, parental characteristics, magnetic resonance brain images, and various test scores for 100 six-year-old children. Through prior studies that construct causal diagrams to explore relationship among variables \citep[e.g.,][]{andrews2018scoring}, seven variables are considered in the image casual mediation analysis, in which {\it delivery mode} is the treatment, {\it cerebrospinal fluid's $($CSF's$)$ MRI} or {\it gray matter's MRI} or {\it white matter's MRI} is the mediator, {\it child's IQ score} is the outcome and {\it gender, maternal age, oxytocin}, and {\it family type} are the pre-treatment covariates. Detailed description of these variables are shown in Table \ref{ta:2}.
	
	\begin{table}[h]
		\centering
		\caption{Description of variables in real data analysis.}
		\begin{tabular}{m{2.5cm}<{\centering}m{0.5cm}<{\centering}m{3.8cm}<{\centering}m{5.4cm}<{\raggedright}m{2cm}<{\centering}}
			\toprule
			\multicolumn{2}{c}{Variable} & Name & Value & Sample size  \\
			\midrule
			\multirow{2}{2.5cm}{\centering Treatment} &\multirow{2}{0.5cm}{\centering $X$}& \multirow{2}*{Delivery mode}& $1$: Caesarean section & \multirow{2}*{93}\\ 
			~&~&~&$0$: Spontaneous vaginal delivery&\\
			\midrule
			\multirow{3}{2cm}{\centering Mediator}&\multirow{3}{0.5cm}{\centering$\bm M$} & CSF's MRI & \multirow{3}*{$\mathcal{R}^{N_1 \times N_2 \times N_3}$}&\multirow{3}*{93} \\ 
			~&~&Gray matter's MRI&~&\\
			~&~&White matter's MRI&~&\\
			\midrule
			{\centering Outcome}&{\centering$Y$} & {\centering IQ score}  &{\centering$(0,\infty)$}& {\centering93}\\ 
			\midrule
			\multirow{8}{2.5cm}{\centering Pre-treatment covariates}&$Z_1$ & \multirow{2}*{\centering Child's gender} &$1$: Boy & \multirow{2}*{\centering 93}\\ 
			~&~&~&$0$: Girl&\\
			~&$Z_2$&{\centering Maternal age}&{\centering$(0,\infty)$}&{\centering93}\\ 
			~&\multirow{2}*{\centering$Z_3$}& \multirow{2}*{\centering Oxytocin}&$1$: Yes& \multirow{2}*{93}\\ 
			~&~& ~&$0$: No&\\
			~&\multirow{3}*{$Z_{41}$}& \multirow{3}*{Family type-1}&$1$: Family only consisting of& \multirow{2}*{93}\\ 
			~&~&~&\ \quad children and their parents&\\
			~&~&~&$0$: Others &\\
			~&\multirow{3}*{$Z_{42}$}& \multirow{3}*{Family type-2}&$1$: Family consisting of children,& \multirow{2}*{93}\\ 
			~&~&~&\ \quad parents and grandparents&\\
			~&~&~&$0$: Others &\\
			\midrule
			~& \textbf{1}.&\multicolumn{3}{l}{The MRIs here are all pre-processed, which are the output images}\\
			~&~ & \multicolumn{3}{l}{of the raw MRIs after four steps' standardization processing by FSL}\\
			~& ~&\multicolumn{3}{l}{Programs (bet2 (first step), FAST (third step)) and AFNI Programs}\\
			~& ~&\multicolumn{3}{l}{ (auto\_warp.py (second step), 3dNwarpApply (fourth step)) based on}\\
			Note:& ~&\multicolumn{3}{l}{the brain image in Figure \ref{fig:S.2}.}\\
			~& \textbf{2}.&\multicolumn{3}{l}{The MRI data of each individual is divided into three classes}\\
			~& ~&\multicolumn{3}{l}{based on CSF, Gray matter and White matter. The value of each locus}\\
			~& ~&\multicolumn{3}{l}{in each class of image is the percentage of that class of tissue}\\
			~& ~&\multicolumn{3}{l}{present at that locus.}\\
			\bottomrule
		\end{tabular} 
		\label{ta:2}
	\end{table}
	
	Our goal is to determine whether mother’s delivery mode has an impact on child’s intellectual performance. Assumptions that are made in Section \ref{sec:2} need to be assessed to qualify the plausibility of the ICMA approach in this application. For Assumption \ref{Ass:1}, in our real data neither the MRI nor the IQ test score of an individual subject should be affected by the treatment of any other subject in the study, thus validating the SUTVA. Assumption \ref{Ass:2} corresponds to the consistency assumption in causal inference, and it is widely used in causal mediation analysis \citep[e.g.,][]{imai2010general,imai2010identification}. Assumptions \ref{Ass:3} (a) and \ref{Ass:5}(a) correspond to the conditional independence and conditional mean independence between treatment and potential outcomes and potential mediators. The data set we use contains a large number of covariates which include variables from all related areas. By constructing causal diagrams to explore the relationship among variables, we found that only two variables (family type, maternal age) affect treatment, and included them as pre-treatment variables in our mediation analysis. So the conditional independence and conditional mean independence can be reasonably assumed. For Assumptions 3 (b) and 5(b), since $\widehat{\bm{\varepsilon}}$ and $\hat\epsilon$ being correlated is one of the most common violations of them, which correspond to the presence of unobserved variables that confound the relationship between $\bm M$ and $Y$, a ball covariance test of independence proposed by \cite{pan2020ball} is peformed to show the uncorrelation between $\widehat{\bm\varepsilon}$ and $\hat{\epsilon}$. Since linear additive models, which correspond to Assumptions 4 and 5 (d), are commonly used in the analysis of neuroimaging data \citep[e.g.,][]{zhou2013tensor,li2018tucker,yu2022mapping,llosa2022reduced}, their usage in this particular application appears reasonable. Additionally, the Kolmogorov-Smirnov statistic is used to test whether $\hat{\epsilon}$ is normal, which is to support the bootstrap step in Section \ref{sec:2.2.2}.
	
	To achieve our goal, model $(\ref{eq:12})$ is fitted and the T-test is peformed. The results show that caesarean section has a statistically significant negative impact on child’s intellectual performance since the estimate of total effect is $-4.5481$ with $p$-value $0.02$ (less than $0.05$). Next, we explore whether this effect is mediated and which loci transmit it by applying the ICMA method to each of the three types of image mediators, respectively. Since the original image data is of dimension $144 \times 192 \times 160$, which is much larger than the sample size $n=93$,  directly applying the ICMA method to the original image data may cause computational difficulties. Moreover, in the theories of tensor regression that under the assumption that the intrinsic tensor rank is known, the order of $\bm{M}$ is restricted by the sample size to achieve the corresponding estimation error in various methods \citep[e.g.,][]{han2022optimal,chen2019non, tomioka2013convex}. Thus, the mediation analysis for the real data here is performed in two steps.
	
	In the first step, we aim to select some potential regions with moderate size. To this end, the original image is transformed into a $\widetilde{N}_1 \times \widetilde{N}_2 \times \widetilde{N}_3$ tensor via average pooling or max pooling, that is, the value of each locus of the new tensor is the average value or max value of the corresponding $\widetilde{\widetilde{N}}_1 \times \widetilde{\widetilde{N}_2} \times \widetilde{\widetilde{N}}_3$ cubic in the original image, see Figure \ref{fig:S.3} for details, where $\widetilde{\widetilde{N}}_d = N_d/\widetilde{N}_d$, for $d = 1,2,3$. Specifically, consider three settings of $(\widetilde{N}_1,\widetilde{N}_2,\widetilde{N}_3)$: $(9,12,10)$, $(16,16,16)$, and $(18,24,20)$. Under each setting the average pooling and max pooling methods are applied, respectively. Next, the ICMA method is applied with the $\widetilde{N}_1 \times \widetilde{N}_2 \times \widetilde{N}_3$ tensor as a `pooled' image mediator to select the potential regions. Specifically, the regions with $p$-values less than 0.2 are taken as the potential regions and proceed to the second step. Additionally, see Appendix C for a discussion about the sensitivity of different settings of $\widetilde{N}_1 \times \widetilde{N}_2 \times \widetilde{N}_3$ and pooling methods to the potential regions selection.

To further explore the indirect effect through the MRI, the second step applies the ICMA method with each potential region that are selected by the first step being an image mediator. It is found that two regions of the frontotemporal area in the white matter have at least 5\% of loci conveying statistically significant mediation effects after adjusting for $p$-values. The results are summarized in Table \ref{ta:3} and Figure \ref{fig:4}. Additionally, the estimate of the direct effect is negative with its $p$-value, calculated by $2- 2\Phi(\frac{\hat{\delta}}{\widehat{\sigma}_{\delta}})$ where $\widehat{\sigma}_{\delta}$ is obtained similarly to $\widehat{\sigma}_{\alpha_{p_1,p_2,p_3}}$, less than 0.05 when each one of the two regions in Figure \ref{fig:4} is considered as an image mediator, respectively.

\begin{table}[h]
	\centering
	\caption{Results of the second step.}
	\label{ta:3}
	\begin{tabular}{m{3cm}<{\centering}m{12cm}<{\centering}}
			\toprule
			Tissue category & Number of regions with at least 5\% of loci conveying statistically significant effects  \\ 
			\midrule
			CSF	& 0\\ 
			\midrule
			Gray matter & 0\\
			\midrule
			White matter & 2\\ 
			\bottomrule
	\end{tabular}
\end{table}

\begin{figure}[h]
	\centering
	\includegraphics[width=0.6\linewidth]{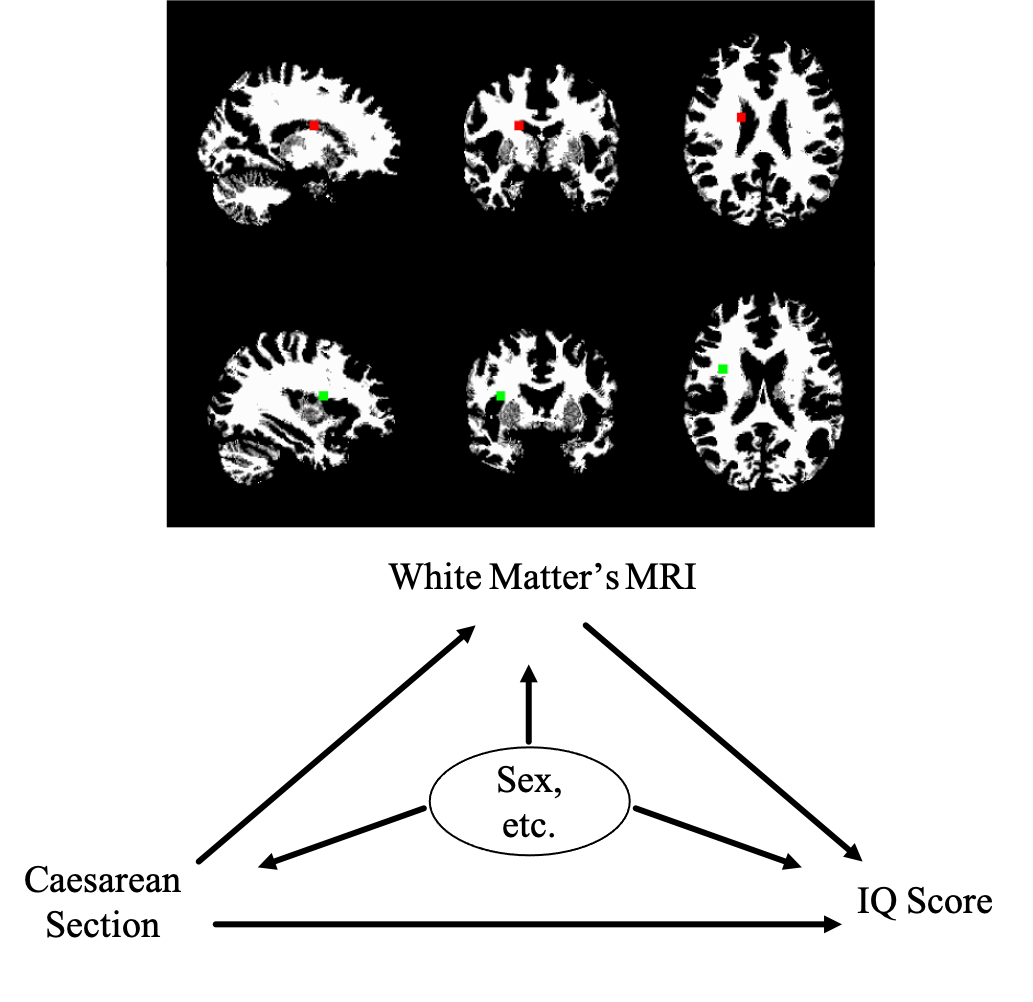}
	\caption{The causal diagram with the white matter's MRI as the mediator. The first row of the image plot shows a slice map of the red region, which is the locus in the $18 \times 24 \times 20$ mediator obtained by average pooling with at least 5\% of loci conveying statistically significant indirect effects in the second step. The second row shows a slice map of the green region, which is the locus obtained by max pooling with at least 5\% of loci conveying statistically significant indirect effects in the second step. Observe that these two regions are located very close to each other and both belong to the frontotemporal area.}
	\label{fig:4}
\end{figure}

A body of literature has detected the influences of the gray and white matter on the development of child's IQ \citep[e.g.,][]{reiss1996brain,wilke2003bright,northam2011total,isaacs2010impact,luby2016breastfeeding}.  The results of our real data analysis indicate that caesarean section has a negative direct effect on child's IQ, and suggest that the two regions of white matter belonging to the frontal-temporal areas in Figure \ref{fig:4} convey indirect effects from the caesarean section to child's IQ development. These findings help address the controversial issues described earlier in this section. Additionally, since limited information is available surrounding the role of white matter in shaping cognitive abilities in children \citep{muetzel2015white}, our findings provide inspiration for future prospective and longitudinal studies.

\section{Discussion}
This article introduces the ICMA method to study image mediation effects. It defines the average casual effects under the potential outcome framework, examines suﬀicient conditions for the valid identification, and develops techniques for estimation and inference. Different from most brain mediation analysis literature that treat summary measures of brain regions as mediator, in this study the image data is directly modeled as the mediator, and pre-treatment covariates are considered.

Note that the LISEMs in this article assumes linearity. However, they can be easily extended to the generalized linear form, that is,
\begin{equation}\label{eq:16}
	\begin{array}{ll}
		\bm{\mathcal A}_1(\mu_{\bm{M} \mid X = x, \mathbb{Z} = z}) = \bm{\eta}_{1} + x* \bm{\alpha} + \sum\limits_{j=1}^{J} z_j * \bm{\Psi}_j,\\
		\mathcal{A}_2(\mu_{Y \mid X = x, \bm{M} = \bm m, \mathbb{Z} = z}) = \eta_2 + x\gamma + \langle\bm{m},\bm{\beta}\rangle + \sum\limits_{j=1}^{J} z_j s_j,
	\end{array}
\end{equation}
where $\bm{\mathcal A}_1$ and $\mathcal{A}_2$ are known link functions, $\mu_{\bm{M} \mid X = x, \mathbb{Z} = z} = E(\bm{M} \mid X = x, \mathbb{Z} = z)$, and $\mu_{Y \mid X = x, \bm{M} = \bm m, \mathbb{Z} = z} = E(Y \mid X = x, \bm{M} = \bm m, \mathbb{Z} = z)$. In this formulation, under Assumption \ref{Ass:1}-\ref{Ass:3}, the total effect $\tau$ can be expressed as
\begin{equation*}
	\begin{aligned}
		\tau &= \gamma(x) + \delta(1-x)\\
		&=  \int\limits_{\mathcal{Z}}\int\limits_{\mathcal{M}}\left\{E(Y \mid X = 1, \bm M = \bm m, \mathbb{Z}=z) - E(Y \mid X = 0, \bm M=\bm m, \mathbb{Z}=z)\right\}d F_{\bm{M} \mid X=x, \mathbb{Z} = z}(\bm m) d F_{\mathbb{Z}}(z)\\
		&+ \int\limits_{\mathcal{Z}}\int\limits_{\mathcal{M}} E(Y \mid X=1-x, \bm M= \bm m, \mathbb{Z}=z)\left\{d F_{\bm{M} \mid X=1, \mathbb{Z}=z}(\bm m)-d F_{\bm M \mid X=0, \mathbb{Z}=z}(\bm m)\right\} d F_{\mathbb{Z}}(z)\\
		&= \int\limits_{\mathcal{Z}}\int\limits_{\mathcal{M}}\{\mathcal{B}_2(\eta_2 + \gamma + \langle\bm{m},\bm{\beta}\rangle + \sum\limits_{j=1}^{J} z_j s_j) - \mathcal{B}_2(\eta_2 + \langle\bm{m},\bm{\beta}\rangle + \sum\limits_{j=1}^{J} z_j s_j)\}d F_{\bm{M} \mid X=x, \mathbb{Z} = z}(\bm m) d F_{\mathbb{Z}}(z)\\
		&+\int\limits_{\mathcal{Z}}\int\limits_{\mathcal{M}}\mathcal{B}_2(\eta_2 + (1-x)\gamma + \langle\bm{m},\bm{\beta}\rangle + \sum\limits_{j=1}^{J} z_j s_j)\left\{d F_{\bm{M} \mid X=1, \mathbb{Z}=z}(\bm m)-d F_{\bm M \mid X=0, \mathbb{Z}=z}(\bm m)\right\} d F_{\mathbb{Z}}(z),
	\end{aligned}
\end{equation*}
where $\mathcal B_2 = \mathcal A_2^{-1}$. Specially, if $\mathcal{B}_2$ is the identity function, then
\begin{equation*}
	\begin{aligned}
		\tau &= \int\limits_{\mathcal{Z}}\int\limits_{\mathcal{M}}\{\eta_2 + \gamma + \langle\bm{m},\bm{\beta}\rangle + \sum\limits_{j=1}^{J} z_j s_j - (\eta_2 + \langle\bm{m},\bm{\beta}\rangle + \sum\limits_{j=1}^{J} z_j s_j)\}d F_{\bm{M} \mid X=x, \mathbb{Z} = z}(\bm m) d F_{\mathbb{Z}}(z)\\
		&+\int\limits_{\mathcal{Z}}\int\limits_{\mathcal{M}}(\eta_2 + (1-x)\gamma + \langle\bm{m},\bm{\beta}\rangle + \sum\limits_{j=1}^{J} z_j s_j)\left\{d F_{\bm{M} \mid X=1, \mathbb{Z}=z}(\bm m)-d F_{\bm M \mid X=0, \mathbb{Z}=z}(\bm m)\right\} d F_{\mathbb{Z}}(z)\\
		& = \gamma + \int\limits_{\mathcal{Z}}\langle\bm{\mathcal B}_1(\bm{\eta}_{1} + \bm{\alpha} + \sum\limits_{j=1}^{J} z_j * \bm{\Psi}_j) - \bm{\mathcal B}_1(\bm{\eta}_{1} + \sum\limits_{j=1}^{J} z_j * \bm{\Psi}_j),\bm{\beta}\rangle d F_{\mathbb{Z}}(z),
	\end{aligned}
\end{equation*}
where $\bm{\mathcal B}_1 = \bm{\mathcal A}_1^{-1}$. Obviously, when $\bm{\mathcal B}_1$ is also the identity function, $\tau = \gamma+ <\bm{\alpha},\bm{\beta}>$.

In this study, only low-dimensional pre-treatment covariates are considered, high-dimensional cases such as the genetic data of children and their parents are not involved. Although there are some ways to deal with high-dimensional covariates \citep[e.g.,][]{luo2017estimating}, these methods may not be effective when the mediator is an image. A recent work \citep{yu2022mapping} proposes a method to map the genetic-imaging-clinical pathway for Alzheimer’s disease, in which the treatment variable is a two-dimensional image, the covariates are ultra-high-dimensional clinical and genetic data, the outcome variable is a continuous scalar, and the relationship among them follow linear structural equation models. However, it is difficult to cope with a three-dimensional image mediator and high-dimensional covariates via this method, since the tensor nuclear norm is generally NP-hard to even approximate \citep{friedland2018nuclear}. The problem of three-dimensional image mediator combined with high-dimensional covariates remains for future work. The post-treatment covariates are not considered here, since in many cases they are affected by treatment, which makes the identification of defined causal effects difficult. A possible solution is to define interventional effects \citep{vanderweele2014effect}, and this issue is left for future work.

\section*{APPENDIX A: PROOF OF IDENTIFICATION}
\subsection*{Proof of Theorem \ref{th:1}}
Under Assumptions \ref{Ass:1}-\ref{Ass:3} following the same argument of the Theorem {\color{red}1} in \citet{imai2010identification}, for $x = 0,1$,
\begin{equation*}
	\begin{aligned}
		\gamma(x)=\int\limits_{\mathcal{Z}}\int\limits_{\mathcal{M}}&\left\{E(Y \mid X = 1, \bm M = \bm m, \mathbb{Z}=z)\right.\\
		&\left.-E(Y \mid X = 0, \bm M=\bm m, \mathbb{Z}=z)\right\}d F_{\bm{M} \mid X=x, \mathbb{Z} = z}(\bm m) d F_{\mathbb{Z}}(z),\\
		\delta(x)=\int\limits_{\mathcal{Z}}\int\limits_{\mathcal{M}} &E(Y \mid X=x, \bm M= \bm m, \mathbb{Z}=z)\\
		&\left\{d F_{\bm{M} \mid X=1, \mathbb{Z}=z}(\bm m)-d F_{\bm M \mid X=0, \mathbb{Z}=z}(\bm m)\right\} d F_{\mathbb{Z}}(z),	
	\end{aligned}
\end{equation*}
where $F_{\mathbb{Z}}(\cdot)$ and $F_{\bm M \mid X, \mathbb{Z}}(\cdot)$ represent the distribution function and the conditional distribution function.

Under Assumption \ref{Ass:4}, 
\begin{equation*}
	\begin{aligned}
		\gamma(x)&=\int\limits_{\mathcal{Z}}\int\limits_{\mathcal{M}}(\eta_{2} +\gamma + \langle\bm m,\bm{\beta}\rangle + \sum\limits_{j=1}^{J} z_j s_j -\eta_{2}\\
		& \qquad \qquad - \langle\bm m,\bm{\beta}\rangle - \sum\limits_{j=1}^{J} z_j s_j)d F_{\bm{M} \mid X=x, \mathbb{Z} = z}(\bm m) d F_{\mathbb{Z}}(z)\\
		&= \gamma.\\
		\delta(x)&=\int\limits_{\mathcal{Z}}\int\limits_{\mathcal{M}}(\eta_{2} +x \gamma + \langle\bm m,\bm{\beta}\rangle + \sum\limits_{j=1}^{J} z_j s_j )\\
		&\qquad \qquad \left\{d F_{\bm{M} \mid X=1, \mathbb{Z}=z}(\bm m)-d F_{\bm M \mid X=0, \mathbb{Z}=z}(\bm m)\right\} d F_{\mathbb{Z}}(z)\\
		& = \int\limits_{\mathcal{Z}}(\langle E(\bm{M} \mid X=1, \mathbb{Z}=z) - E(\bm M \mid X=0, \mathbb{Z}=z),\bm{\beta}\rangle) d F_{\mathbb{Z}}(z)\\
		& = \langle \bm \alpha, \bm \beta \rangle.\\
		\tau &= \gamma(1) + \delta(0) = \gamma(0) + \delta(1) = \gamma + \langle \bm \alpha, \bm \beta \rangle.
	\end{aligned}
\end{equation*}
\subsection*{Proof of Lemma \ref{le:1} and Theorem \ref{th:2}}
\begin{lemma}\label{le:1}
	Assumption $\ref{Ass:5}$ is weaker than Assumptions $\ref{Ass:3}$ and $\ref{Ass:4}$.
\end{lemma}
First, note that Assumption \ref{Ass:3} (a) implies\\
$$\begin{array}{rclr}
	\bm{M}_i(x) &\perp& X_i \mid \mathbb{Z}_i = z & \qquad (\ast)\\
	Y_i(x^{\prime},\bm{m}) &\perp& X_i \mid \mathbb{Z}_i = z & \qquad (\ast\ast)\\
	Y_i(x^{\prime},\bm{m}) &\perp& X_i \mid \bm{M}_i(x) = \bm{m}^{\prime}, \mathbb{Z}_i = z &\qquad (\ast\ast\ast)
\end{array}$$\\

According to $(\ast)$, $E(\bm{M}(x)\mid X = x, \mathbb{Z} = z) = E(\bm{M}(x) \mid \mathbb{Z} = z)$. 

By $(\ast\ast\ast)$, 
\begin{equation*}
	\begin{aligned}
		E&(Y(x,\bm{M}(x))\mid X = x, \bm{M}(x) = \bm{m}, \mathbb{Z} = z) \\
		&\ = E(Y(x,\bm m)\mid X = x, \bm{M}(x) = \bm{m}, \mathbb{Z} = z)\\
		&\overset{(\ast\ast\ast)}{=} E(Y(x,\bm m)\mid \bm{M}(x) = \bm{m}, \mathbb{Z} = z)\\
		&\ = E(Y(x,\bm{M}(x))\mid \bm{M}(x) = \bm{m}, \mathbb{Z} = z).
	\end{aligned}	
\end{equation*}

Then under Assumption \ref{Ass:3} (b), $(\ast\ast)$ and $(\ast\ast\ast)$, we obtain
\begin{equation*}
	\begin{aligned}
		E&(Y(x,\bm{M}(x))\mid X = x^{\prime}, \bm{M}(x^{\prime}) = \bm{m}, \mathbb{Z} = z) \\
		&\overset{2 (b)}{=} E(Y(x,\bm m)\mid X = x^{\prime}, \mathbb{Z} = z)\\
		&\overset{(\ast\ast)}{=} E(Y(x,\bm m)\mid \mathbb{Z} = z),\\
		E&(Y(x,\bm{m})\mid X = x^{\prime}, \bm{M}(x^{\prime}) = \bm{m}, \mathbb{Z} = z) \\
		&\overset{(\ast\ast\ast)}{=} E(Y(x,\bm m)\mid \bm{M}(x^{\prime}) = \bm{m}, \mathbb{Z} = z), \\
		E&(Y(x,\bm m)\mid \bm{M}(x^{\prime}) = \bm{m}, \mathbb{Z} = z) = E(Y(x,\bm m)\mid \mathbb{Z} = z).
	\end{aligned}
\end{equation*}

Additionally, obviously, Assumption \ref{Ass:5} (c) can be deduced by $0 < Pr(X = x\mid \mathbb{Z}_i = z)$ and $0 < p(\bm M(x) = m \mid X = x,\mathbb{Z} = z)$.

Further, under Assumption \ref{Ass:3}, 
\begin{equation*}
	\begin{aligned}
		E(\bm{M}(x) \mid \mathbb{Z} = z)&=E(\bm{M}(x) \mid X=x, \mathbb{Z} = z)\\
		&=E(\bm{M} \mid X=x, \mathbb{Z} = z).\\
		E(Y(x,\bm{m}) \mid \mathbb{Z} = z)&=E(Y(x,\bm{m}) \mid X = x, \mathbb{Z} = z)\\
		&= \int\limits_{\mathcal{M}} E(Y(x, \bm m)) \mid X = x, \bm{M}(x) = \bm m, \mathbb{Z} = z) d F_{\bm{M}(x) \mid X = x, \mathbb{Z} = z}(\bm m)\\
		&= \int\limits_{\mathcal{M}} E(Y \mid X = x, \bm{M}(x) = \bm m, \mathbb{Z} = z) d F_{\bm{M}\mid X = x, \mathbb{Z} = z}(\bm m).\\
	\end{aligned}
\end{equation*}
Then by Assumption \ref{Ass:4}, 
\begin{equation*}
	\begin{aligned}
		E(\bm{M}(x) \mid \mathbb{Z} = z) &= \bm{\eta}_{1} + x * \bm{\alpha} + \sum\limits_{j=1}^{J} z_j* \bm{\Psi}_{j},\\
		E(Y(x,\bm{m}) \mid \mathbb{Z} = z) &= \eta_{2} + x \gamma + \langle\bm{m},\bm{\beta}\rangle + \sum\limits_{j=1}^{J} z_j s_j.
	\end{aligned}
\end{equation*}

Thus, Assumptions \ref{Ass:3} and \ref{Ass:4} imply Assumption \ref{Ass:5}. However, in general, the mean independence between two random variables does not imply they are independent of each other. Thus the inverse does not hold. Hence, Assumption $\ref{Ass:5}$ is weaker than Assumptions \ref{Ass:3} and \ref{Ass:4}.

Next, we prove Theorem \ref{th:2}.

Under Assumptions \ref{Ass:1}, \ref{Ass:2} and \ref{Ass:5}, 
\begin{equation*}
	\begin{aligned}
		E(\bm{M} \mid X = x, \mathbb{Z} = z) &= E(\bm{M}(X) \mid X = x, \mathbb{Z} = z)\\
		& = E(\bm{M}(x) \mid X = x, \mathbb{Z} = z)\\
		& = E(\bm{M}(x) \mid \mathbb{Z} = z)\\
		&= \bm{\eta}_1^{(c)} + x * \bm{\alpha}^{(c)} + \sum\limits_{j=1}^{J}z_j * \bm{\Psi}_j^{(c)}.\\
		E(Y \mid X = x, \bm{M}(X) = \bm m, \mathbb{Z} = z) &= E(Y(X, \bm M(X)) \mid X = x, \bm{M}(X) = \bm m, \mathbb{Z} = z)\\
		& = E(Y(x, \bm M(x)) \mid X = x, \bm{M}(x) = \bm m, \mathbb{Z} = z)\\
		& = E(Y(x, \bm M(x)) \mid \bm{M}(x) = \bm m, \mathbb{Z} = z)\\
		& = E(Y(x, \bm m) \mid \bm{M}(x) = \bm m, \mathbb{Z} = z)\\
		& = E(Y(x, \bm m) \mid \mathbb{Z} = z)\\
		& = \eta_2^{(c)} + x\gamma^{(c)} + \langle\bm{m}, \bm{\beta}^{(c)}\rangle + \sum\limits_{j=1}^{J}z_j s_j^{(c)}.
	\end{aligned}
\end{equation*}

Comparing the results above with the LISEMs, then
\[
\begin{aligned}
	(\bm{\eta}_1^{(c)}-\bm{\eta}_1) + x * (\bm{\alpha}^{(c)} - \bm{\alpha}) + \sum\limits_{j=1}^{J}z_j * (\bm{\Psi}_j^{(c)}-\bm{\Psi}) = \bm{0},\\
	(\eta_2^{(c)} - \eta_2) + x(\gamma^{(c)} - \gamma) + \langle\bm{m}, \bm{\beta}^{(c)}- \bm{\beta}^{(c)}\rangle + \sum\limits_{j=1}^{J}z_j (s_j^{(c)}-s_j) = 0.
\end{aligned}
\]

Since $x$, $\bm m$, $z$ are arbitrary, then $\bm{\alpha}^{(c)} = \bm{\alpha}, \ \gamma^{(c)} = \gamma, \ \bm{\beta}^{(c)} = \bm{\beta}$. This follows that $E(\bm M(1) - \bm M(0)) = \bm{\alpha}$, $E(Y(1, \bm m) - Y(0, \bm m)) = \gamma$ and $E(Y(x, \bm m) - Y(x, \bm m^{\prime})) = \langle \bm{m}-\bm{m}^{\prime}, \bm \beta \rangle$.

Further, under Assumptions \ref{Ass:1}, \ref{Ass:2} and \ref{Ass:5},
\[
\begin{aligned}
	E\left(Y\left(x, \boldsymbol{M}\left(x^{\prime}\right)\right) \mid \mathbb{Z}=z\right) &=\int_{\mathcal{M}} E\left(Y\left(x, \boldsymbol{M}\left(x^{\prime}\right)\right) \mid \boldsymbol{M}\left(x^{\prime}\right)=\boldsymbol{m}, \mathbb{Z}=z\right) d F_{\boldsymbol{M}\left(x^{\prime}\right) \mid \mathbb{Z}=z}(\boldsymbol{m}) \\
	&=\int_{\mathcal{M}} E(Y(x, \boldsymbol{m}) \mid \boldsymbol{M}\left(x^{\prime}\right)=\boldsymbol{m}, \mathbb{Z}=z) d F_{\boldsymbol{M}\left(x^{\prime}\right) \mid \mathbb{Z}=z}(\boldsymbol{m}) \\
	&=\int_{\mathcal{M}} E(Y(x, \boldsymbol{m}) \mid \mathbb{Z}=z) d F_{\boldsymbol{M}\left(x^{\prime}\right) \mid \mathbb{Z}=z}(\boldsymbol{m}) \\
	&=\int_{\mathcal{M}}(\eta_{2}^{(c)}+x \gamma^{(c)}+\langle\boldsymbol{m}, \boldsymbol{\beta}^{(c)}\rangle+\sum_{j=1}^{J} z_{j} s_{j}^{(c)}) d F_{\boldsymbol{M}\left(x^{\prime}\right) \mid \mathbb{Z}=z}(\boldsymbol{m}) \\
	&=\eta_{2}^{(c)}+x \gamma^{(c)}+\langle E\left(\boldsymbol{M}\left(x^{\prime}\right) \mid \mathbb{Z}=z\right), \boldsymbol{\beta}^{(c)}\rangle+\sum_{j=1}^{J} z_{j} s_{j} \\
	&=\eta_{2}^{(c)}+x \gamma^{(c)}+\langle\boldsymbol{\eta}_{1}^{(c)}+x^{\prime} * \boldsymbol{\alpha}^{(c)}+\sum_{j=1}^{J} z_{j} * \boldsymbol{\Psi}_{j}^{(c)}, \boldsymbol{\beta}^{(c)}\rangle+\sum_{j=1}^{J} z_{j} s_{j}^{(c)} \\
	&=\eta_{2}^{(c)}+\langle\boldsymbol{\eta}_{1}^{(c)}, \boldsymbol{\beta}^{(c)}\rangle+x \gamma^{(c)}+x^{\prime}\langle\boldsymbol{\alpha}^{(c)}, \boldsymbol{\beta}^{(c)}\rangle+\sum_{j=1}^{J} z_{j}\left(\langle\boldsymbol{\Psi}_{j}^{(c)}, \boldsymbol{\beta}^{(c)}\rangle+s_{j}^{(c)}\right).
\end{aligned}
\]

Hence, $\tau = \gamma + \langle\bm{\beta},\bm \alpha\rangle, \gamma(x) = \gamma, \delta(x) = \langle\bm{\beta},\bm \alpha\rangle$.

\section*{APPENDIX B: ESTIMATION STRATEGY FOR TENSOR PARAMETER IN MODEL ($\ref{eq:11}$)}
\setcounter{equation}{0}
\renewcommand{\theequation}{\rm{S}.\arabic{equation}}
To begin, we briefly review some notations in matrix/array operations. Given two matrices $U = [u_1, \cdots,u_p] \in \mathcal{R}^{k \times p} $ and $V = [v_1, \cdots, v_q] \in \mathcal{R}^{l \times q} $, the {\it Kronecker product} is the {\it kl}-by-{\it pq} matrix $U \otimes V = [u_1 \otimes V,  u_2 \otimes V, \cdots, u_p \otimes V]$. The {\it mode}-{\it d matricization}, $V(d)$, maps a tensor $\bm V \in \mathcal{R}^{N_1 \times \cdots \times N_D}$ into a $N_d \times \prod\limits_{d^{\prime} \neq d}^{D}N_{d^{\prime}}$ matrix such that the $(p_1, . . . , p_D)$ element of the tensor $\bm V$ maps to the $(p_d, l)$ element of the matrix $V(d)$, where $l =1+\sum_{d^{\prime} \neq d}\left(p_{d^{\prime}}-1\right) \prod_{d^{\prime \prime}<d^{\prime}, d^{\prime \prime} \neq d} N_{d^{\prime \prime}}$. The {\it mode}-{\it d multiplication} of the tensor $\bm V$ with a matrix $U \in \mathcal{R}^{N_d \times q}$, denoted by $\bm V \times_d U \in \mathcal{R}^{N_1\times \cdots \times q \times \cdots N_D}$, is the multiplication of the mode-{\it d} fibers of $\bm V$ by $U$, that is, $U V(d)$.
	
	Next, model $(\ref{eq:11})$ can be rewritten as
	\begin{equation}\label{eq:S1}
		E(Y_i \mid Z_i, X_i, \bm{M}_i)= B_{X_i,Z_i}\vartheta+\left\langle \bm{M}_i, \bm{B} \times_{1} U_1 \times_{2} U_2 \times_{3} U_3\right\rangle,
	\end{equation}
	by assuming that the coefficient tensor $\bm \beta$ follows a Tucker decomposition, that is, $\bm \beta = \bm{B} \times_{1} U_1 \times_{2} U_2 \times_{3} U_3$, where $\bm B \in \mathcal{R}^{R_1\times R_2 \times R_3}$ is the core tensor, $U_d$ is the $N_d$-by-$R_d$ factor matrice and `$\times_d$' represents the tensor-matrix product along mode $d$, for $d = 1, 2 ,3$. Additionally, $B_{X_i,Z_i} = (X_i,Z_{i1},\cdots,Z_{iJ})$ and $\vartheta = (\delta,s_1,\cdots,s_J)^{T}$. In fact, model $(\ref{eq:S1})$ is a special case of model ({\color{red}4}) in \cite{li2018tucker} when assuming that $Y$ belongs to a gaussian family with probability density,
	\begin{equation}\label{eq:S2}p(y\mid \mu,\sigma)=\frac{1}{\sigma \sqrt{2 \pi}} e^{-\frac{1}{2}\left(\frac{y-\mu}{\sigma}\right)^2}.\end{equation}
	Thus, given $n$ iid data $\{(z_i, x_i,\bm{m}_i,y_i), i = 1, . . . , n\}$, the loglikelihood function for $(\ref{eq:S2})$ is
	\begin{equation}\label{eq:S3}
		\begin{array}{rl}
			\ell(\vartheta,U_1,U_2,U_3,\bm{B}) &= -\frac{n}{2} \ln (2 \pi)-\frac{n}{2} \ln \sigma_i^2-\frac{1}{2 \sigma_i^2} \sum_{i=1}^n\left(y_i-\mu_i\right)^2\\
			& = \sum_{i = 1}^{n}\frac{y_i\mu_i - \mu_i^2/2}{\sigma_i^2} + \sum^{n}_{i = 1} -(\frac{y_i^2}{2\sigma_i^{2}} + \frac{ln(2\pi)}{2}+\frac{ln(\sigma_i^2)}{2}),
		\end{array}
	\end{equation}
    where $\mu_i$ is related to regression parameters $(\vartheta,U_1,U_2,U_3,\bm{B})$ through $(\ref{eq:S1})$.
	
	\cite{li2018tucker} suggests a block ascent algorithm to obtain the maximum likelihood estimation (MLE) for the Tucker tensor regression model $(\ref{eq:S1})$, we briefly review it as follows.
	\begin{algorithm}
		\caption{\leftline{\large Block relaxation algorithm for fitting the Tucker tensor regression.}}
		\begin{algorithmic}
			\State {\small Initialize: $\vartheta^{(0)} = \arg \max_{\vartheta}\ell(\vartheta, 0, \cdots , 0)$, $U(d)^{(0)}$ a random matrix for $d = 1, \cdots, D$, and $\bm{B}^{(0)} \in \mathcal{R}^{R_1\times R_2 \times R_3}$ a random tensor.}
			\Repeat
			\For{d = 1, 2, 3}
			\State $U_d^{(t+1)}=\arg\max_{U_{d}} \ell(\vartheta^{(t)}, U_1^{(t+1)}, \ldots, U_{d-1}^{(t+1)}, U_d, U_{d+1}^{(t)}, \ldots, U_D^{(t)}, \bm{B}^{(t)})$
			\EndFor
			\State $\bm{B}^{(t+1)}=\arg\max_{\bm{B}} \ell\left(\vartheta^{(t)}, U_1^{(t+1)}, \ldots, U_D^{(t+1)}, \bm{B}\right)$
			\State $\vartheta^{(t+1)}=\arg\max_{\vartheta} \ell\left(\vartheta, U_1^{(t+1)}, \ldots, U_D^{(t+1)}, \bm{B}^{(t+1)}\right)$
			\Until $\ell\left(\vartheta^{(t+1)}\right)-\ell\left(\vartheta^{(t)}\right)<\varsigma$
		\end{algorithmic}
	\end{algorithm}
	
	Additionally, in practical application, a warm initialization may cause a substantial reduction in computational burden and a more stable estimate, so in the Section \ref{sec:4} of this paper and the additional simulation in Appendix C, the initial value $\bm{\beta}^{(0)}$ is obtained by the sequentially truncated HOSVD (ST-HOSVD) \citep{vannieuwenhoven2012new} based on $(\ref{eq:S.3})$.
	
	\begin{proposition}\label{pro:1}
		Suppose $\widehat{\bm \Delta} = (\widehat{\bm \alpha}, \widehat{\bm \Psi}_1, \cdots, \widehat{\bm \Psi}_J)$ is the ordinary least square estimator of $\bm \Delta = (\bm \alpha, \bm \Psi_1, \cdots, \bm \Psi_J)$ of equation $(\ref{eq:10})$, and $\hat{\theta}$ is the ordinary least square estimator of $\theta$ of equation $(\ref{eq:12})$, then equation $(\ref{eq:11})$ can be rewritten as
		\begin{equation}\label{eq:S.3}
			\widehat \omega = A_{\widehat{\varepsilon}} \beta + \hat{\epsilon},
		\end{equation}
		where $ \widehat \omega = Y - B_{X,Z}\hat{\theta} $, $\beta = \mathrm{vec}(\bm \beta)$, $\hat{\epsilon} = \Lambda\epsilon$, and $A_{\widehat{\varepsilon}} = (\widehat{\varepsilon}_1, \cdots, \widehat{\varepsilon}_n)^{T}$, $\widehat{\varepsilon}_i = \mathrm{vec}(\widehat{\bm \varepsilon}_i)$, $\widehat{\bm \varepsilon}_i = \bm M_i - X_i * \widehat{\bm \alpha} - \sum\limits_{j=1}^{J} Z_{ij} * \widehat{\bm{\Psi}}_{j}$, $\Lambda=(I - B_{X,Z}^{T}(B_{X,Z}^{T}B_{X,Z})^{-1}B_{X,Z})$.
	\end{proposition}
	\begin{proof}
		Note that $\widehat{\Delta} = (A_{X,Z}^{T}A_{X,Z})^{-1}A_{X,Z}^{T}M$, then $\widehat{\varepsilon} = (I - A_{X,Z}(A_{X,Z}^{T}A_{X,Z})^{-1}A_{X,Z}^{T})\varepsilon$.
		
		Due to $A_{X,Z} = B_{X,Z} \otimes I$, $\varepsilon = \mathrm{vec}(A_{\varepsilon}^{T})$, $\widehat{\varepsilon}= \mathrm{vec}(A_{\widehat \varepsilon}^{T})$ and the Properties of Kronecker product, then $\mathrm{vec}(A_{\widehat \varepsilon}^{T}) = (I - (B_{X,Z} \otimes I)((B_{X,Z} \otimes I)^{T}(B_{X,Z} \otimes I))^{-1}(B_{X,Z}^T \otimes I)) \mathrm{vec}(A_{\varepsilon}^{T}) = ((I - B_{X,Z}(B_{X,Z}^{T}B_{X,Z})^{-1}B_{X,Z}^T) \otimes I) \mathrm{vec}(A_{\varepsilon}^{T})$, that is, $\mathrm{vec}(A_{\widehat \varepsilon}^{T}) = \mathrm{vec}(A_{\varepsilon}^{T}(I - B_{X,Z}(B_{X,Z}^{T}B_{X,Z})^{-1}B_{X,Z}^T))$. Thus, $A_{\widehat{\varepsilon}} = (I - B_{X,Z}(B_{X,Z}^{T}B_{X,Z})^{-1}B_{X,Z}^{T})A_{\varepsilon}$.
		
		Therefore, 
		\[
		\begin{array}{l}
			\hat \omega = Y - B_{X,Z}\hat{\theta}\\
			\quad = (I - B_{X,Z}(B_{X,Z}^TB_{X,Z})^{-1}B_{X,Z}^{T})A_{\varepsilon}\beta + (I - B_{X,Z}(B_{X,Z}^TB_{X,Z})^{-1}B_{X,Z}^{T}) \epsilon\\
			\quad = A_{\widehat \varepsilon}\beta + \hat \epsilon.
		\end{array}
		\]
	\end{proof}
	
	At last, we discuss the selection of Tucker rank $R$ for $\bm{\beta}$. In general, tuning rank is challenging for tensors as many parameter values need to be tuned simultaneously \citep{hillar2013most}. Specifically, inspired by \cite{li2018tucker}, which treat this as a model selection problem, we adopt Bayesian information criterion(BIC), $-2 \log \ell+\log (n) p_e$, to choose the Tucker rank for the tensor parameter from a `warm' set. Here, $\ell$ is defined in $(\ref{eq:S3})$, $p_e = \sum_{d=1}^{3} N_d R_d+\prod_{d=1}^{3} R_d-\sum_{d=1}^{3} R_d^2$ and the warm set is selected by considering the sample size $n$ and the order of $\bm{M}$. Table {\color{red}1} of \cite{han2022optimal} shows the requirement for the sample size to achieve the corresponding estimation error in various methods \citep[e.g.,][]{han2022optimal,chen2019non, tomioka2013convex}. Based on this result, we choose $R = (R_1,R_2,R_3)$ from $\{(1, 1,1),(1,2,1),(2,2,1),(2,2,2)\}$ in the simulation part of this paper by considering the sample size $n$ = 100 and $\bm{M}$ is of type $8\times8\times8$. In the real data analysis part $R$ is chosen from $\{(1, 1,1),(1,2,1),(2,2,1),(2,2,2),(2,3,2),(2,3,3),(3,3,3)\}$ since the sample size is 93 and the order of $\bm{M}$ in each direction ranges from 8 to 24.

\section*{APPENDIX C: DISCUSSIONS ABOUT THE SENSITIVITY OF DIFFERENT $\tilde{N}_1 \times \tilde{N}_2 \times \tilde{N}_3 $ AND POOLING METHODS TO POTENTIAL REGIONS SELECTION}
\setcounter{table}{0}
\renewcommand{\thetable}{\rm{S}.\arabic{table}}
Here, the sensitivity of different $\tilde{N}_1 \times \tilde{N}_2 \times \tilde{N}_3 $ and pooling methods to potential regions selection is discussed. Figure \ref{fig:4} demonstrates that the selection of potential regions is not sensitive to the different pooling methods. However, the selection of potential regions is inconsistent for different sizes of $\tilde{N}_1 \times \tilde{N}_2 \times \tilde{N}_3$ for this particular real data application. Specifically, three settings of $\tilde{N}_1 \times \tilde{N}_2 \times \tilde{N}_3$ are considered: (9, 12, 10), (16, 16, 16), and (18, 24, 20), for each setting both the average pooling and max pooling methods are used. Only with (18, 24, 20) does the procedure select regions that are of statistical significance, which is shown in Figure 4. For the other two settings the procedure doesn't find any regions with at least 5\% of loci conveying statistically significant indirect effects in the second step. One reason may be that the estimates of $\epsilon$ in the first step are not normally distributed by the Kolmogorov-Smirnov test, causing some inconsistency of the bootstrap selection procedure. To overcome this problem, our strategy is to set multiple choices of pooling cube sizes and the threshold p-value to be 0.2 to include more potential mediating regions in the first step, then in the second step to perform a more strict procedure to locate the specific mediating loci.

To assess the performance of the proposed method in more cases, additional simulation is performed. The setting of this simulation is summarized in Table \ref{ta:S.1}. The results in Table \ref{ta:S.2} and Figure \ref{fig:S.5} show that all estimates of the mediating area cover the true region, which demonstrates that our approach is robust to the choice of $\tilde{N}_1 \times \tilde{N}_2 \times \tilde{N}_3$ and pooling methods in this setting.

\begin{table}[h]
	\centering
	\caption{Settings of the additional simulation.}
	\label{ta:S.1}
	\begin{tabular}{m{2cm}<{\centering}m{11cm}<{\centering}m{2cm}}
		\toprule
		Variable & Generation formula & Value\\ 
		\midrule
		$X$ 	& sample from $\mathcal{B}(1,\pi_i)$& 0 \text{or} 1\\ 
		$\bm{M}$ & $X_i*\bm{\alpha} + Z_i*\bm\Psi + \bm{\varepsilon}_i$& $\mathcal{R}^{144\times 196 \times 160}$\\
		$Y$  & $X_i \gamma + Z_i s + \langle \bm{M}_i, \bm{\beta} \rangle + \epsilon_{i}$&$\mathcal{R}$\\ 
		\midrule
		Note: & Here, $Z_i \sim \mathcal{U}(0,2)$, $\text{logit}(\pi_i) = 0.5 - 0.5Z_i$, $\bm{\Psi}$ is random with $\mathrm{vec}(\bm{\Psi}) \sim \mathcal{N}(0,I)$, $\epsilon_i \sim \mathcal{N}(0,0.1^2)$, $\delta = s = 10$, and the non-zero loci of $\bm \alpha$, $\bm{\varepsilon}_i$, and $\bm \beta$ are restricted in the same $8 \times 8 \times 8$ region. Moreover, the elements of $\bm{\varepsilon}_i$ in this $8 \times 8 \times 8$ region follow a normal distribution with mean zero and standard deviation 0.3.& \\
		\bottomrule
	\end{tabular}
\end{table}

\begin{table}[h]
	\centering
	\caption{Results of the additional simulation.}
	\label{ta:S.2}
	\begin{tabular}{m{3cm}<{\centering}m{3cm}<{\centering}m{4cm}<{\centering}m{6cm}<{\centering}}
			\toprule
			Pooling method & Order of the pooled mediator & Regions with loci having mediation effects in the the pooled mediators& Centre of these loci in the original image\\ 
			\midrule
			\makecell{\\Max pooling\\ \\ \\ Average pooling\\} & \makecell{$9 \times 12 \times 10$ \\$16 \times 16 \times 16$ \\ $18 \times 24 \times 20$ \\ $9 \times 12 \times 10$ \\$16 \times 16 \times 16$ \\ $18 \times 24 \times 20$}&
			\makecell{(5, 7, 6)\\(8, 9, 9), (9, 9, 9)\\(9  13  11), (10, 13, 11)\\(5, 7, 6)\\(8, 9, 9), (9, 9, 9)\\(9  13  11), (10,13,11)} & \makecell{(69.5, 100.3, 85.8)\\(69.1, 100.1, 85.8), (73, 101, 86.7)\\( 69.0, 100.3, 85.8), (73, 100.8, 86)\\(69.5, 100.3, 85.8)\\(69.1, 100.1, 85.8), (73, 101, 86.7)\\( 69.0, 100.3, 85.8), (73, 100.8, 86)} \\ 
			\bottomrule
	\end{tabular}
\end{table}

\section*{APPENDIX D: SUPPLEMENTARY FIGURES FOR SIMULATION STUDIES AND REAL DATA ANALYSIS}
\setcounter{figure}{0}

\renewcommand{\thefigure}{\rm{S}.\arabic{figure}}
\begin{figure}[h]
	\centering
	\includegraphics[width=1\linewidth]{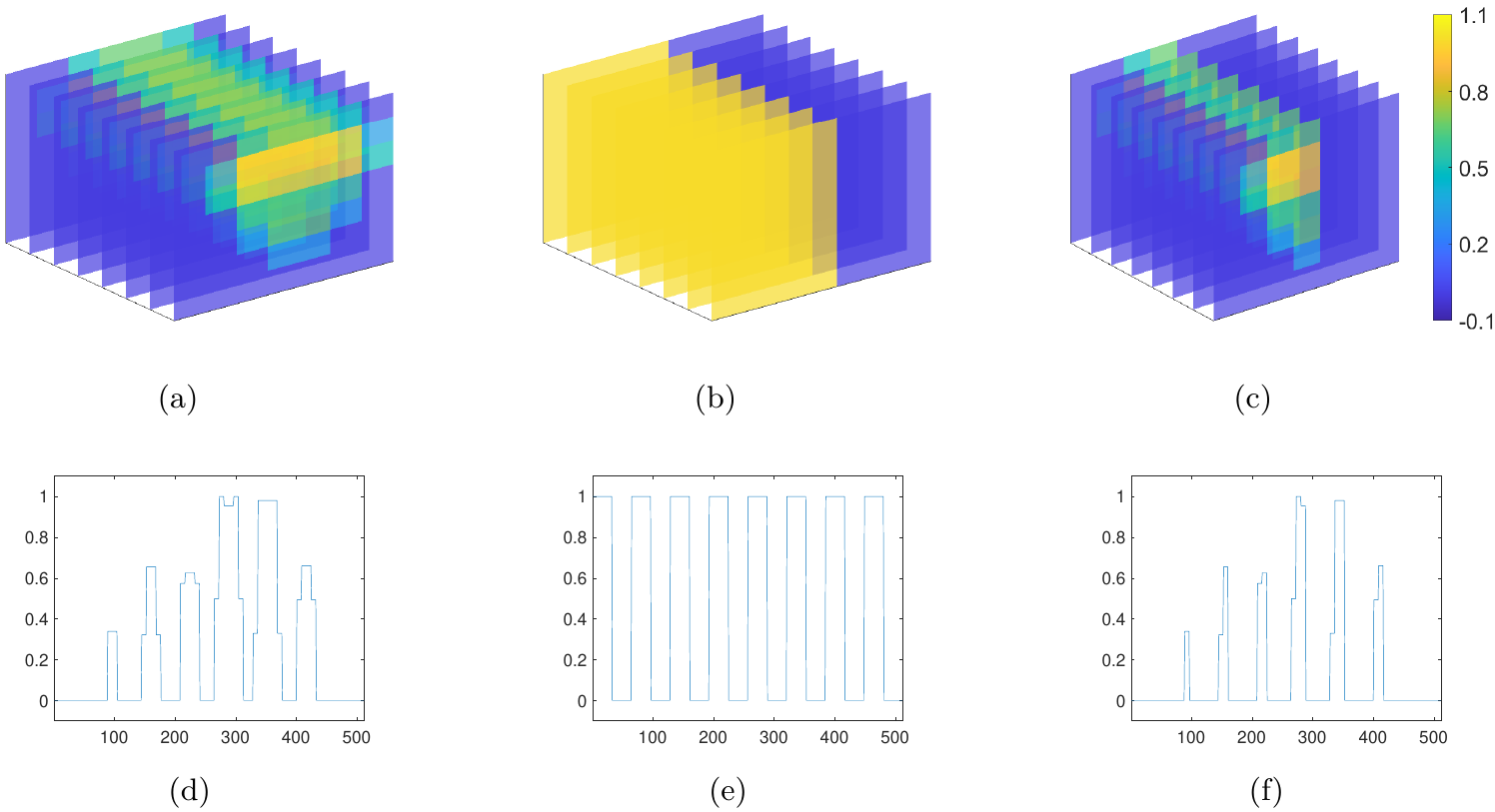}
	\caption{(a) represents the slice image of $\bm \alpha$ and $\bm \Psi$, and the non-zero part of it is a mixed-color elliptical cylinder. (b) represents the slice image of $\bm \beta$, which is generated by two cubes of different colors. (c) represents the slice image of $\bm{\alpha\beta}$. The colors represent the grid point values. (d) represents the plots of $\mathrm{vec}(\bm \alpha)$ or $\mathrm{vec}(\bm \Psi)$, and (e), (f) represent the plots of $\mathrm{vec}(\bm \beta)$ and $\mathrm{vec}(\bm{\alpha\beta})$.}
	\label{fig:S.1}
\end{figure}
\begin{figure}[htbp]
	\centering
	\includegraphics[width=1\linewidth]{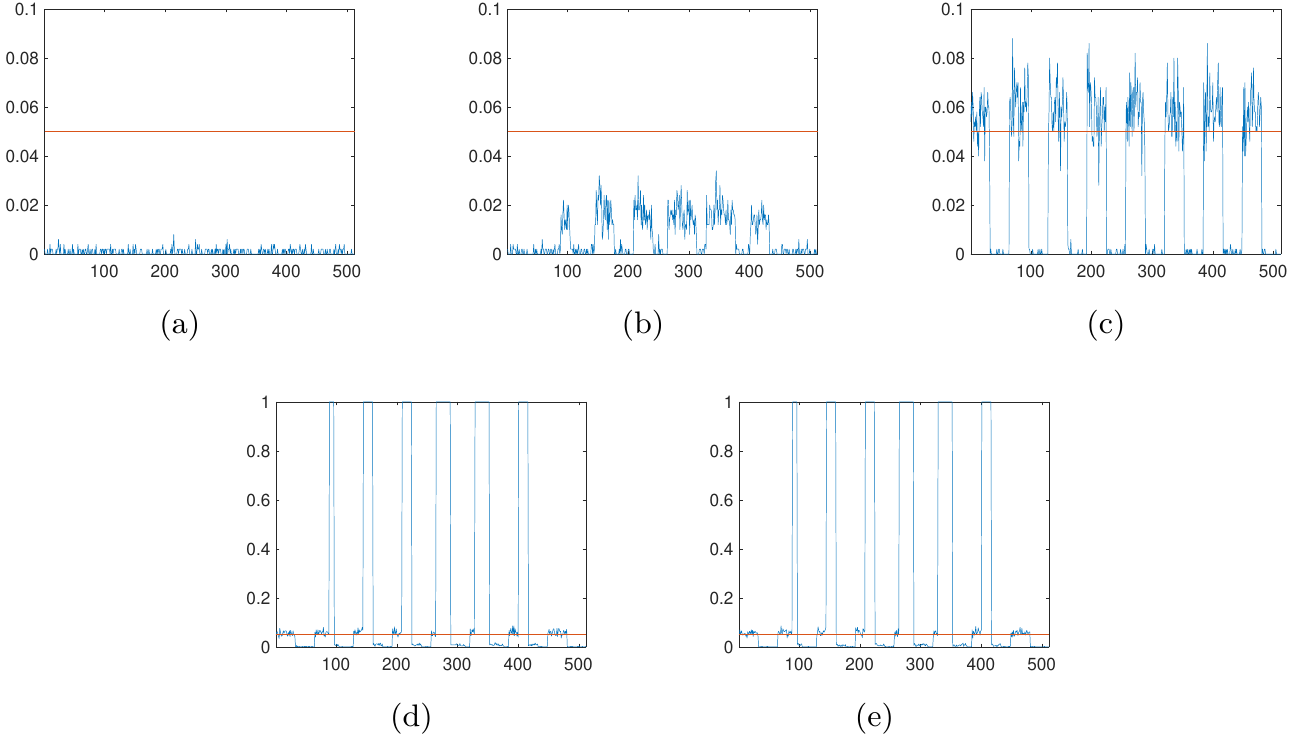}
	\caption{Results of the first three scenarios are shown in (a), (b), and (c), respectively. Each plot shows the proportion of times that $\bm{\alpha} \circ \bm{\beta}$ is statistically significant (empirical ${\it p}<0.05$) among the 500 replications for each locus. The results illustrate that the method provides adequate control of the false positive rate in all three simulated scenarios. Results of the last two scenarios are shown in (d) and (e), respectively. Each plot illustrates the power of the method in detecting true positives in the loci of $\alpha_{p_1,p_2,p_3} \beta_{p_1,p_2,p_3} \neq 0$, while appropriately controlling for false positives in the loci of $\alpha_{p_1,p_2,p_3}\beta_{p_1,p_2,p_3} = 0$. The $x$-axis of all plots represents the straightened loci.}
	\label{fig:S.2}
\end{figure}
\begin{figure}[h]
	\centering
	\includegraphics[width=1\linewidth]{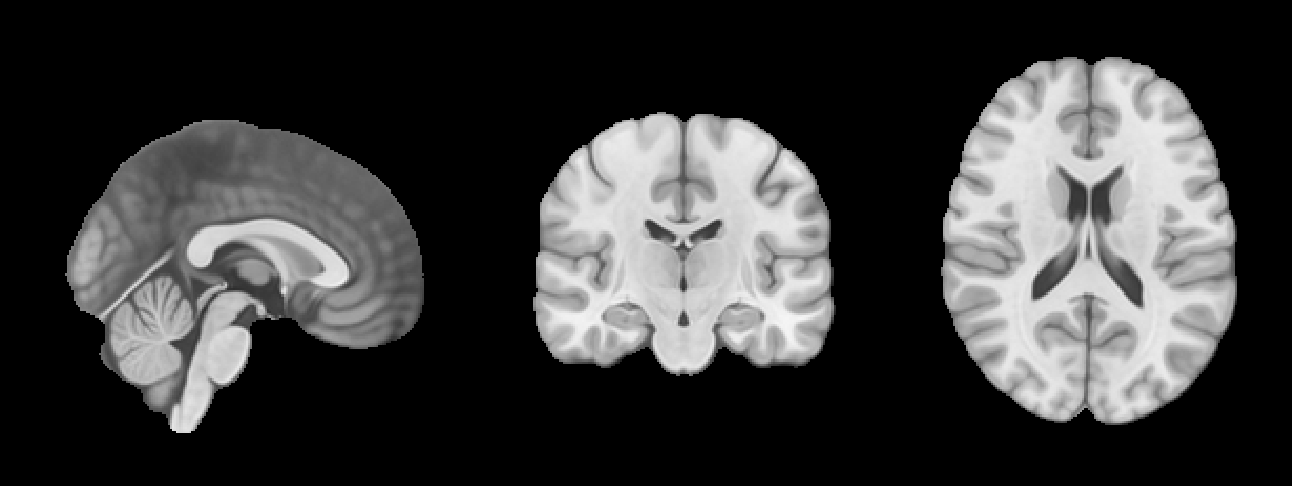}
	\caption{The MNI152 standard-space T1-weighted average structural template image.}
	\label{fig:S.3}
\end{figure}
\begin{figure}[h]
	\centering
	\includegraphics[width=1\linewidth]{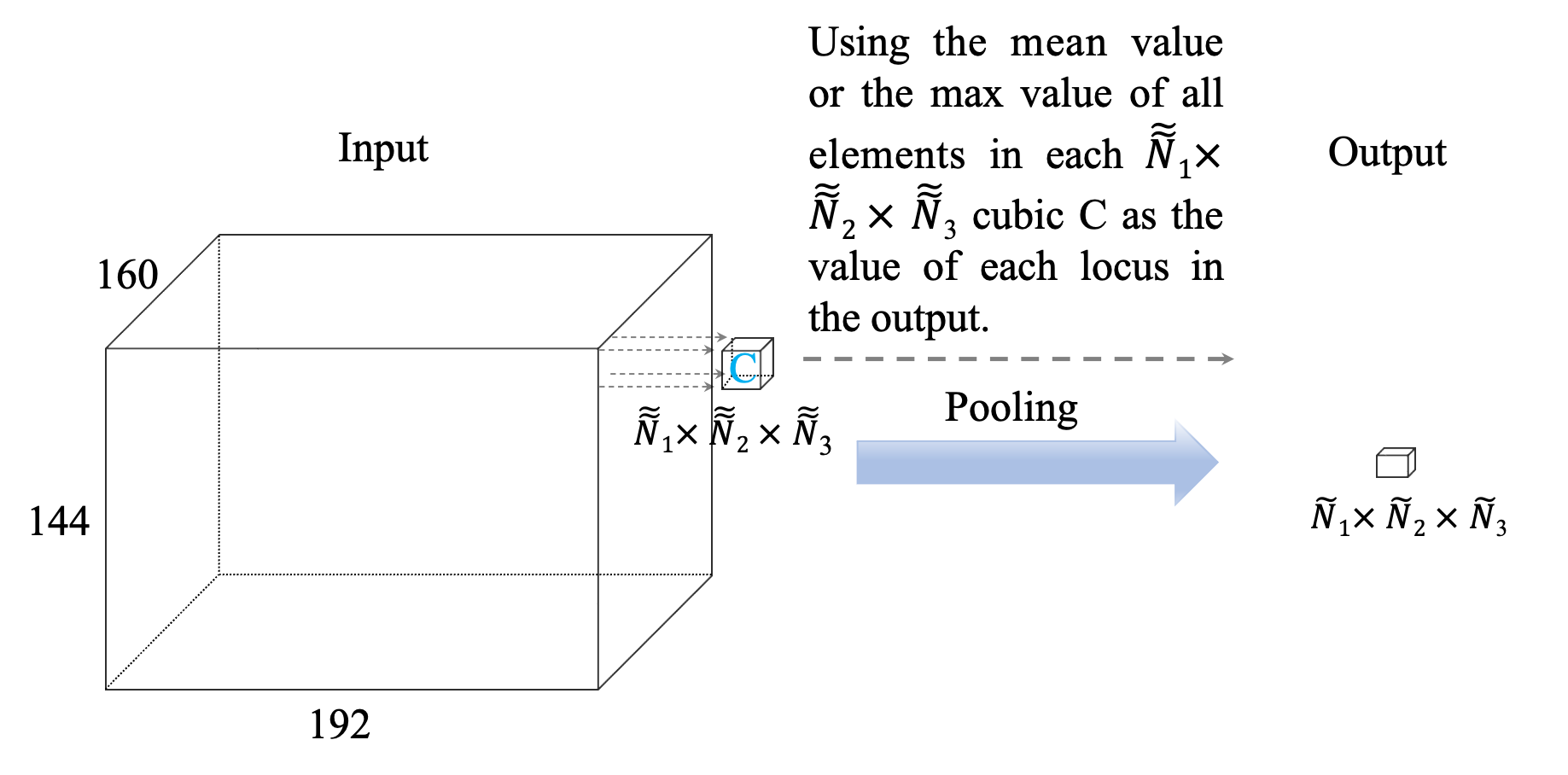}
	\caption{\colorbox{yellow}{Pooling.}
	\label{fig:S.4}}
\end{figure}
\begin{figure}[h]
	\centering
	\includegraphics[width=1\linewidth]{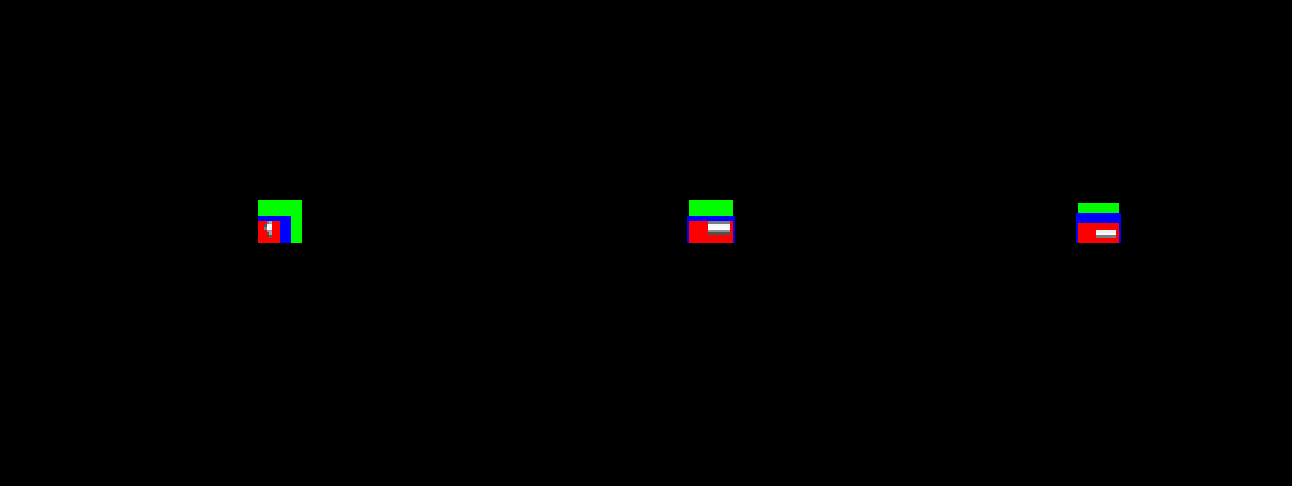}
	\caption{The true indirect effect region and its estimates. Here, the mediation analysis strategy proposed in Section \ref{sec:4} is performed. The red box is a $8\times 8 \times 8$ region, the blue box is a $9\times 12 \times 10$ region, the green box is a $16 \times 16 \times 16$ region, and the true region is gray and white.  All the regions cover the whole intermediary area, which shows the robustness of the method to the choice of $\tilde{N}_1 \times \tilde{N}_2 \times \tilde{N}_3$.}
	\label{fig:S.5}
\end{figure}

\newpage
\section*{ACKNOWLEDGEMENTS}
This work was supported by the National Natural Science Foundation of China (No.11771146, No.11831008), the Basic Research Project of Shanghai Science and Technology Commission (No.19JC1410101), and the ``Flower of Happiness'' Fund Pilot Project of East China Normal University (No.2019JK2203).

\section*{DATA AVAILABILITY STATEMENT}
The data supporting the findings of this study can be obtained by application to the Institute of Brain Science, East China Normal University. The data are not publicly available due to privacy or ethics or restrictions.

\section*{REFERENCES}
\bibliographystyle{apalike}
\begingroup
\renewcommand{\section}[2]{}
\bibliography{wenxian}

\begin{thebibliography}{}

\bibitem[Ahmed et~al., 2020]{ahmed2020tensor}
Ahmed, T., Raja, H., and Bajwa, W.~U. (2020).
\newblock Tensor regression using low-rank and sparse tucker decompositions.
\newblock {\em SIAM Journal on Mathematics of Data Science}, 2(4):944--966.

\bibitem[Albert, 2008]{albert2008mediation}
Albert, J.~M. (2008).
\newblock Mediation analysis via potential outcomes models.
\newblock {\em Statistics in medicine}, 27(8):1282--1304.

\bibitem[Andrews et~al., 2018]{andrews2018scoring}
Andrews, B., Ramsey, J., and Cooper, G.~F. (2018).
\newblock Scoring bayesian networks of mixed variables.
\newblock {\em International journal of data science and analytics},
  6(1):3--18.

\bibitem[Angrist et~al., 1996]{angrist1996identification}
Angrist, J.~D., Imbens, G.~W., and Rubin, D.~B. (1996).
\newblock Identification of causal effects using instrumental variables.
\newblock {\em Journal of the American statistical Association},
  91(434):444--455.

\bibitem[Avin et~al., 2005]{avin2005identifiability}
Avin, C., Shpitser, I., and Pearl, J. (2005).
\newblock Identifiability of path-specific effects.
\newblock In {\em Proceedings of the 19th international joint conference on
  Artificial intelligence}, pages 357--363. Morgan Kaufman, San Francisco, CA.

\bibitem[Baron and Kenny, 1986]{baron1986moderator}
Baron, R.~M. and Kenny, D.~A. (1986).
\newblock The moderator--mediator variable distinction in social psychological
  research: Conceptual, strategic, and statistical considerations.
\newblock {\em Journal of personality and social psychology}, 51(6):1173.

\bibitem[Benjamini and Hochberg, 1995]{benjamini1995controlling}
Benjamini, Y. and Hochberg, Y. (1995).
\newblock Controlling the false discovery rate: a practical and powerful
  approach to multiple testing.
\newblock {\em Journal of the Royal statistical society: series B
  (Methodological)}, 57(1):289--300.

\bibitem[Blake et~al., 2021]{blake2021association}
Blake, J.~A., Gardner, M., Najman, J., and Scott, J.~G. (2021).
\newblock The association of birth by caesarean section and cognitive outcomes
  in offspring: a systematic review.
\newblock {\em Social Psychiatry and Psychiatric Epidemiology}, 56(4):533--545.

\bibitem[Blazkova et~al., 2020]{blazkova2020impact}
Blazkova, B., Pastorkova, A., Solansky, I., Veleminsky, M., Rossnerova, A.,
  Honkova, K., Rossner, P., and Sram, R.~J. (2020).
\newblock The impact of cesarean and vaginal delivery on results of
  psychological cognitive test in 5 year old children.
\newblock {\em Medicina}, 56(10):554.

\bibitem[Caffo et~al., 2008]{caffo2008brain}
Caffo, B., Chen, S., Stewart, W., Bolla, K., Yousem, D., Davatzikos, C., and
  Schwartz, B.~S. (2008).
\newblock Are brain volumes based on magnetic resonance imaging mediators of
  the associations of cumulative lead dose with cognitive function?
\newblock {\em American journal of epidemiology}, 167(4):429--437.

\bibitem[Casey et~al., 2005]{casey2005imaging}
Casey, B., Tottenham, N., Liston, C., and Durston, S. (2005).
\newblock Imaging the developing brain: what have we learned about cognitive
  development?
\newblock {\em Trends in cognitive sciences}, 9(3):104--110.

\bibitem[Chen et~al., 2019]{chen2019non}
Chen, H., Raskutti, G., and Yuan, M. (2019).
\newblock Non-convex projected gradient descent for generalized low-rank tensor
  regression.
\newblock {\em The Journal of Machine Learning Research}, 20(1):172--208.

\bibitem[Dai et~al., 2020]{dai2020multiple}
Dai, J.~Y., Stanford, J.~L., and LeBlanc, M. (2020).
\newblock A multiple-testing procedure for high-dimensional mediation
  hypotheses.
\newblock {\em Journal of the American Statistical Association}, pages 1--16.

\bibitem[Daniel et~al., 2015]{daniel2015causal}
Daniel, R.~M., De~Stavola, B.~L., Cousens, S., and Vansteelandt, S. (2015).
\newblock Causal mediation analysis with multiple mediators.
\newblock {\em Biometrics}, 71(1):1--14.

\bibitem[Derkach et~al., 2019]{derkach2019high}
Derkach, A., Pfeiffer, R.~M., Chen, T.-H., and Sampson, J.~N. (2019).
\newblock High dimensional mediation analysis with latent variables.
\newblock {\em Biometrics}, 75(3):745--756.

\bibitem[Friedland and Lim, 2018]{friedland2018nuclear}
Friedland, S. and Lim, L.-H. (2018).
\newblock Nuclear norm of higher-order tensors.
\newblock {\em Mathematics of Computation}, 87(311):1255--1281.

\bibitem[Giedd and Rapoport, 2010]{Giedd2010structural}
Giedd, J.~N. and Rapoport, J.~L. (2010).
\newblock Structural mri of pediatric brain development: what have we learned
  and where are we going?
\newblock {\em Neuron}, 67(5):728--734.

\bibitem[Guo et~al., 2022]{guo2022statistical}
Guo, X., Li, R., Liu, J., and Zeng, M. (2022).
\newblock Statistical inference for linear mediation models with
  high-dimensional mediators and application to studying stock reaction to
  covid-19 pandemic.
\newblock {\em Journal of Econometrics}.

\bibitem[Haier et~al., 2004]{haier2004structural}
Haier, R.~J., Jung, R.~E., Yeo, R.~A., Head, K., and Alkire, M.~T. (2004).
\newblock Structural brain variation and general intelligence.
\newblock {\em Neuroimage}, 23(1):425--433.

\bibitem[Han et~al., 2022]{han2022optimal}
Han, R., Willett, R., and Zhang, A.~R. (2022).
\newblock An optimal statistical and computational framework for generalized
  tensor estimation.
\newblock {\em The Annals of Statistics}, 50(1):1--29.

\bibitem[Hillar and Lim, 2013]{hillar2013most}
Hillar, C.~J. and Lim, L.-H. (2013).
\newblock Most tensor problems are np-hard.
\newblock {\em Journal of the ACM (JACM)}, 60(6):1--39.

\bibitem[Hoff, 2011]{hoff2011separable}
Hoff, P.~D. (2011).
\newblock Separable covariance arrays via the tucker product, with applications
  to multivariate relational data.
\newblock {\em Bayesian Analysis}, 6(2):179--196.

\bibitem[Imai et~al., 2010a]{imai2010general}
Imai, K., Keele, L., and Tingley, D. (2010a).
\newblock A general approach to causal mediation analysis.
\newblock {\em Psychological methods}, 15(4):309.

\bibitem[Imai et~al., 2010b]{imai2010identification}
Imai, K., Keele, L., and Yamamoto, T. (2010b).
\newblock Identification, inference and sensitivity analysis for causal
  mediation effects.
\newblock {\em Statistical science}, 25(1):51--71.

\bibitem[Isaacs et~al., 2010]{isaacs2010impact}
Isaacs, E.~B., Fischl, B.~R., Quinn, B.~T., Chong, W.~K., Gadian, D.~G., and
  Lucas, A. (2010).
\newblock Impact of breast milk on intelligence quotient, brain size, and white
  matter development.
\newblock {\em Pediatric research}, 67(4):357--362.

\bibitem[James and Brett, 1984]{james1984mediators}
James, L.~R. and Brett, J.~M. (1984).
\newblock Mediators, moderators, and tests for mediation.
\newblock {\em Journal of applied psychology}, 69(2):307.

\bibitem[Jiang and Colditz, 2023]{jiang2023causal}
Jiang, S. and Colditz, G.~A. (2023).
\newblock Causal mediation analysis using high-dimensional image mediator
  bounded in irregular domain with an application to breast cancer.
\newblock {\em Biometrics}.

\bibitem[Kressner and Perisa, 2017]{kressner2017recompression}
Kressner, D. and Perisa, L. (2017).
\newblock Recompression of hadamard products of tensors in tucker format.
\newblock {\em SIAM Journal on Scientific Computing}, 39(5):A1879--A1902.

\bibitem[Li et~al., 2011]{li2011cesarean}
Li, H.-T., Ye, R.-W., Pei, L.-J., Ren, A.-G., Zheng, X.-Y., and Liu, J.-M.
  (2011).
\newblock Cesarean delivery on maternal request and childhood intelligence: a
  cohort study.

\bibitem[Li et~al., 2018]{li2018tucker}
Li, X., Xu, D., Zhou, H., and Li, L. (2018).
\newblock Tucker tensor regression and neuroimaging analysis.
\newblock {\em Statistics in Biosciences}, 10:520--545.

\bibitem[Lindquist, 2012]{lindquist2012functional}
Lindquist, M.~A. (2012).
\newblock Functional causal mediation analysis with an application to brain
  connectivity.
\newblock {\em Journal of the American Statistical Association},
  107(500):1297--1309.

\bibitem[Liu et~al., 2019]{liu2019neuronal}
Liu, Y., Schubert, J., Sonnenberg, L., Helbig, K.~L., Hoei-Hansen, C.~E., Koko,
  M., Rannap, M., Lauxmann, S., Huq, M., Schneider, M.~C., et~al. (2019).
\newblock Neuronal mechanisms of mutations in scn8a causing epilepsy or
  intellectual disability.
\newblock {\em Brain}, 142(2):376--390.

\bibitem[Liu et~al., 2022]{liu2022large}
Liu, Z., Shen, J., Barfield, R., Schwartz, J., Baccarelli, A.~A., and Lin, X.
  (2022).
\newblock Large-scale hypothesis testing for causal mediation effects with
  applications in genome-wide epigenetic studies.
\newblock {\em Journal of the American Statistical Association},
  117(537):67--81.

\bibitem[Llosa-Vite and Maitra, 2022]{llosa2022reduced}
Llosa-Vite, C. and Maitra, R. (2022).
\newblock Reduced-rank tensor-on-tensor regression and tensor-variate analysis
  of variance.
\newblock {\em IEEE Transactions on Pattern Analysis and Machine Intelligence},
  45(2):2282--2296.

\bibitem[Luby et~al., 2016]{luby2016breastfeeding}
Luby, J.~L., Belden, A.~C., Whalen, D., Harms, M.~P., and Barch, D.~M. (2016).
\newblock Breastfeeding and childhood iq: The mediating role of gray matter
  volume.
\newblock {\em Journal of the American Academy of Child \& Adolescent
  Psychiatry}, 55(5):367--375.

\bibitem[Luo et~al., 2017]{luo2017estimating}
Luo, W., Zhu, Y., and Ghosh, D. (2017).
\newblock On estimating regression-based causal effects using sufficient
  dimension reduction.
\newblock {\em Biometrika}, 104(1):51--65.

\bibitem[Luo and Zhang, 2021]{luo2021low}
Luo, Y. and Zhang, A.~R. (2021).
\newblock Low-rank tensor estimation via riemannian gauss-newton: Statistical
  optimality and second-order convergence.
\newblock {\em arXiv preprint arXiv:2104.12031}.

\bibitem[MacKinnon et~al., 2002]{mackinnon2002comparison}
MacKinnon, D.~P., Lockwood, C.~M., Hoffman, J.~M., West, S.~G., and Sheets, V.
  (2002).
\newblock A comparison of methods to test mediation and other intervening
  variable effects.
\newblock {\em Psychological methods}, 7(1):83.

\bibitem[Miles et~al., 2020]{miles2020semiparametric}
Miles, C.~H., Shpitser, I., Kanki, P., Meloni, S., and Tchetgen~Tchetgen, E.~J.
  (2020).
\newblock On semiparametric estimation of a path-specific effect in the
  presence of mediator-outcome confounding.
\newblock {\em Biometrika}, 107(1):159--172.

\bibitem[Muetzel et~al., 2015]{muetzel2015white}
Muetzel, R.~L., Mous, S.~E., van~der Ende, J., Blanken, L.~M., van~der Lugt,
  A., Jaddoe, V.~W., Verhulst, F.~C., Tiemeier, H., and White, T. (2015).
\newblock White matter integrity and cognitive performance in school-age
  children: a population-based neuroimaging study.
\newblock {\em Neuroimage}, 119:119--128.

\bibitem[Northam et~al., 2011]{northam2011total}
Northam, G.~B., Li{\'e}geois, F., Chong, W.~K., S.~Wyatt, J., and Baldeweg, T.
  (2011).
\newblock Total brain white matter is a major determinant of iq in adolescents
  born preterm.
\newblock {\em Annals of neurology}, 69(4):702--711.

\bibitem[Pan et~al., 2020]{pan2020ball}
Pan, W., Wang, X., Zhang, H., Zhu, H., and Zhu, J. (2020).
\newblock Ball covariance: A generic measure of dependence in banach space.
\newblock {\em Journal of the American Statistical Association},
  115(529):307--317.

\bibitem[Pearl, 2001]{pearl2001direct}
Pearl, J. (2001).
\newblock Direct and indirect effects.
\newblock In {\em Proceedings of the Seventeenth conference on Uncertainty in
  artificial intelligence {\rm (J. S. Breese and D. Koller, eds.)}}, pages
  411--420. Morgan Kaufman, San Francisco, CA.

\bibitem[Petersen et~al., 2006]{petersen2006estimation}
Petersen, M.~L., Sinisi, S.~E., and van~der Laan, M.~J. (2006).
\newblock Estimation of direct causal effects.
\newblock {\em Epidemiology}, pages 276--284.

\bibitem[Reiss et~al., 1996]{reiss1996brain}
Reiss, A.~L., Abrams, M.~T., Singer, H.~S., Ross, J.~L., and Denckla, M.~B.
  (1996).
\newblock Brain development, gender and iq in children: a volumetric imaging
  study.
\newblock {\em Brain}, 119(5):1763--1774.

\bibitem[Robins, 2003]{robins2003semantics}
Robins, J.~M. (2003).
\newblock Semantics of causal dag models and the identification of direct and
  indirect effects.
\newblock {\em Oxford Statistical Science Series}, pages 70--82.

\bibitem[Robins and Greenland, 1992]{robins1992identifiability}
Robins, J.~M. and Greenland, S. (1992).
\newblock Identifiability and exchangeability for direct and indirect effects.
\newblock {\em Epidemiology}, pages 143--155.

\bibitem[Rubin, 1974]{rubin1974estimating}
Rubin, D.~B. (1974).
\newblock Estimating causal effects of treatments in randomized and
  nonrandomized studies.
\newblock {\em Journal of educational Psychology}, 66(5):688.

\bibitem[Rubin, 1978]{rubin1978bayesian}
Rubin, D.~B. (1978).
\newblock Bayesian inference for causal effects: The role of randomization.
\newblock {\em The Annals of statistics}, pages 34--58.

\bibitem[Rubin, 1980]{rubin1980randomization}
Rubin, D.~B. (1980).
\newblock Randomization analysis of experimental data: The fisher randomization
  test comment.
\newblock {\em Journal of the American Statistical Association},
  75(371):591--593.

\bibitem[Rubin, 2004]{rubin2004direct}
Rubin, D.~B. (2004).
\newblock Direct and indirect causal effects via potential outcomes.
\newblock {\em Scandinavian Journal of Statistics}, 31(2):161--170.

\bibitem[Sobel, 1982]{sobel1982asymptotic}
Sobel, M.~E. (1982).
\newblock Asymptotic confidence intervals for indirect effects in structural
  equation models.
\newblock {\em Sociological methodology}, 13:290--312.

\bibitem[Sobel, 2008]{sobel2008identification}
Sobel, M.~E. (2008).
\newblock Identification of causal parameters in randomized studies with
  mediating variables.
\newblock {\em Journal of Educational and Behavioral Statistics},
  33(2):230--251.

\bibitem[Tibshirani, 1996]{tibshirani1996regression}
Tibshirani, R. (1996).
\newblock Regression shrinkage and selection via the lasso.
\newblock {\em Journal of the Royal Statistical Society: Series B
  (Methodological)}, 58(1):267--288.

\bibitem[Tomioka and Suzuki, 2013]{tomioka2013convex}
Tomioka, R. and Suzuki, T. (2013).
\newblock Convex tensor decomposition via structured schatten norm
  regularization.
\newblock {\em Advances in neural information processing systems}, 26.

\bibitem[Tucker, 1966]{tucker1966some}
Tucker, L.~R. (1966).
\newblock Some mathematical notes on three-mode factor analysis.
\newblock {\em Psychometrika}, 31(3):279--311.

\bibitem[Van~de Geer et~al., 2014]{van2014asymptotically}
Van~de Geer, S., B{\"u}hlmann, P., Ritov, Y., and Dezeure, R. (2014).
\newblock On asymptotically optimal confidence regions and tests for
  high-dimensional models.
\newblock {\em The Annals of Statistics}, 42(3):1166--1202.

\bibitem[VanderWeele, 2015]{vanderweele2015explanation}
VanderWeele, T. (2015).
\newblock {\em Explanation in causal inference: methods for mediation and
  interaction}.
\newblock Oxford University Press.

\bibitem[VanderWeele and Vansteelandt, 2014]{vanderweele2014mediation}
VanderWeele, T. and Vansteelandt, S. (2014).
\newblock Mediation analysis with multiple mediators.
\newblock {\em Epidemiologic methods}, 2(1):95--115.

\bibitem[VanderWeele et~al., 2014]{vanderweele2014effect}
VanderWeele, T.~J., Vansteelandt, S., and Robins, J.~M. (2014).
\newblock Effect decomposition in the presence of an exposure-induced
  mediator-outcome confounder.
\newblock {\em Epidemiology (Cambridge, Mass.)}, 25(2):300.

\bibitem[Vannieuwenhoven et~al., 2012]{vannieuwenhoven2012new}
Vannieuwenhoven, N., Vandebril, R., and Meerbergen, K. (2012).
\newblock A new truncation strategy for the higher-order singular value
  decomposition.
\newblock {\em SIAM Journal on Scientific Computing}, 34(2):A1027--A1052.

\bibitem[Wang et~al., 2017]{wang2017generalized}
Wang, X., Zhu, H., and Initiative, A. D.~N. (2017).
\newblock Generalized scalar-on-image regression models via total variation.
\newblock {\em Journal of the American Statistical Association},
  112(519):1156--1168.

\bibitem[Wilke et~al., 2003]{wilke2003bright}
Wilke, M., Sohn, J.-H., Byars, A.~W., and Holland, S.~K. (2003).
\newblock Bright spots: correlations of gray matter volume with iq in a normal
  pediatric population.
\newblock {\em Neuroimage}, 20(1):202--215.

\bibitem[Yu et~al., 2022]{yu2022mapping}
Yu, D., Wang, L., Kong, D., and Zhu, H. (2022).
\newblock Mapping the genetic-imaging-clinical pathway with applications to
  alzheimer’s disease.
\newblock {\em Journal of the American Statistical Association},
  117(540):1656--1668.

\bibitem[Zhang and Zhang, 2014]{zhang2014confidence}
Zhang, C.-H. and Zhang, S.~S. (2014).
\newblock Confidence intervals for low dimensional parameters in high
  dimensional linear models.
\newblock {\em Journal of the Royal Statistical Society: Series B (Statistical
  Methodology)}, 76(1):217--242.

\bibitem[Zhang et~al., 2021]{zhang2021mediation}
Zhang, H., Zheng, Y., Hou, L., Zheng, C., and Liu, L. (2021).
\newblock Mediation analysis for survival data with high-dimensional mediators.
\newblock {\em Bioinformatics}, 37(21):3815--3821.

\bibitem[Zhang et~al., 2019]{zhang2019association}
Zhang, T., Sidorchuk, A., Sevilla-Cerme{\~n}o, L., Vilaplana-P{\'e}rez, A.,
  Chang, Z., Larsson, H., Mataix-Cols, D., and de~la Cruz, L.~F. (2019).
\newblock Association of cesarean delivery with risk of neurodevelopmental and
  psychiatric disorders in the offspring: a systematic review and
  meta-analysis.
\newblock {\em JAMA network open}, 2(8):e1910236--e1910236.

\bibitem[Zhang et~al., 2018]{zhang2018impact}
Zhang, Y., Shao, J., Yu, M., and Wang, L. (2018).
\newblock Impact of sufficient dimension reduction in nonparametric estimation
  of causal effect.
\newblock {\em Statistical Theory and Related Fields}, 2(1):89--95.

\bibitem[Zhang et~al., 2020]{zhang2020quantile}
Zhang, Y., Wang, L., Yu, M., and Shao, J. (2020).
\newblock Quantile treatment effect estimation with dimension reduction.
\newblock {\em Statistical Theory and Related Fields}, 4(2):202--213.

\bibitem[Zhao et~al., 2021]{zhao2021multimodal}
Zhao, Y., Li, L., and Caffo, B.~S. (2021).
\newblock Multimodal neuroimaging data integration and pathway analysis.
\newblock {\em Biometrics}, 77(3):879--889.

\bibitem[Zhao and Luo, 2019]{zhao2019granger}
Zhao, Y. and Luo, X. (2019).
\newblock Granger mediation analysis of multiple time series with an
  application to functional magnetic resonance imaging.
\newblock {\em Biometrics}, 75(3):788--798.

\bibitem[Zheng and Zhou, 2015]{zheng2015causal}
Zheng, C. and Zhou, X.-H. (2015).
\newblock Causal mediation analysis in the multilevel intervention and
  multicomponent mediator case.
\newblock {\em Journal of the Royal Statistical Society: Series B (Statistical
  Methodology)}, 77(3):581--615.

\bibitem[Zhou and Li, 2014]{zhou2014regularized}
Zhou, H. and Li, L. (2014).
\newblock Regularized matrix regression.
\newblock {\em Journal of the Royal Statistical Society: Series B (Statistical
  Methodology)}, 76(2):463--483.

\bibitem[Zhou et~al., 2013]{zhou2013tensor}
Zhou, H., Li, L., and Zhu, H. (2013).
\newblock Tensor regression with applications in neuroimaging data analysis.
\newblock {\em Journal of the American Statistical Association},
  108(502):540--552.

\bibitem[Zhou et~al., 2020]{zhou2020estimation}
Zhou, R.~R., Wang, L., and Zhao, S.~D. (2020).
\newblock Estimation and inference for the indirect effect in high-dimensional
  linear mediation models.
\newblock {\em Biometrika}, 107(3):573--589.

\end{thebibliography}
\endgroup

\newpage
\appendix
\setcounter{equation}{0}
\renewcommand{\theequation}{A.\arabic{equation}}

\end{document}